\documentclass[11pt]{article}

\usepackage[english]{babel}

\usepackage[letterpaper,top=2cm,bottom=2cm,left=3cm,right=3cm,marginparwidth=1.75cm]{geometry}

\usepackage{amsmath}
\usepackage{graphicx}
\usepackage[table]{xcolor}
\usepackage[colorlinks=true, allcolors=blue]{hyperref}
\usepackage{hhline}

\usepackage{multirow} 
\usepackage{array}
\usepackage{hhline}
\usepackage{caption}
\usepackage{amsthm, amssymb, bm, graphicx, hyperref, mathrsfs}
\usepackage[ruled]{algorithm2e}
\usepackage[sort,numbers]{natbib}
\usepackage{comment}
\usepackage[most]{tcolorbox}
\usepackage{enumitem}

\usepackage{booktabs}

\usepackage{cleveref}
\newtheorem{theorem}{Theorem}
\newtheorem{claim}{Claim}
\newtheorem{lemma}{Lemma}

\newtheorem{Definition}{Definition}

\newcommand{\OPT}{\textup{\textsf{OPT}}\xspace}
\newcommand{\ALG}{\textup{\textsf{ALG}}\xspace}
\newcommand{\NSJF}{\textup{\textsf{NSJF}}\xspace}
\newcommand{\SRPT}{\textup{\textsf{SRPT}}\xspace}
\newcommand{\SmallOnly}{\textsf{Small-Only}\xspace}
\newcommand{\LargeOnly}{\textsf{Large-Only}\xspace}
\newcommand{\All}{\textsf{Mixed}\xspace}
\newcommand{\LargeSet}{A}

\definecolor{Gray0}{RGB}{160, 160, 160}
\definecolor{Gray1}{RGB}{105, 105, 105}
\definecolor{FixedBlue}{RGB}{68, 114, 196}
\definecolor{LeftOrange}{RGB}{237, 125, 49}
\definecolor{RightTeal}{RGB}{0, 176, 150}

\newtcolorbox{openq}[1]{
    enhanced,
    frame hidden,
    borderline west={2pt}{0pt}{black!70},
    colback=black!4,
    coltitle=black,
    fonttitle=\bfseries,
    title={#1},
    left=8pt, right=8pt,
    top=4pt, bottom=4pt,
    arc=0mm,
    toptitle=4pt, bottomtitle=0pt
}

\newtcolorbox{construction}{                                                                   enhanced,
      frame hidden,
      colback=black!4,
      left=8pt, right=8pt,
      top=4pt, bottom=4pt,
      arc=0mm
  }

\newcommand{\E}{\mathbb{E}}

\usepackage{xcolor}
\usepackage{color-edits}
\addauthor{zonghan}{orange}	
\title{Online Flow Time Minimization:\\
Tight Bounds for Non-Preemptive Algorithms}

\author{
  Yutong Geng\\
  Shanghai Jiao Tong University\\
  \texttt{gengyutong@sjtu.edu.cn}
  \and
  Enze Sun\\
  The University of Hong Kong\\
  \texttt{sunenze@connect.hku.hk}
  \and
  Zonghan Yang\\
  Shanghai Jiao Tong University\\
  \texttt{fstqwq@sjtu.edu.cn}
  \and
  Yuhao Zhang\\
  Shanghai Jiao Tong University\\
  \texttt{zhang\_yuhao@sjtu.edu.cn}
}

\date{}
\begin{document}
\maketitle



\begin{abstract}
This paper studies the online scheduling problem of minimizing total flow time for $n$ jobs on $m$ identical machines. A classical $\Omega(n)$ lower bound shows that no deterministic single-machine algorithm can beat the trivial greedy, even when $n$ is known in advance. However, this barrier is specific to deterministic algorithms on a single machine, leaving open what randomization, multiple machines, or the kill-and-restart capability can achieve.

We give a nearly complete answer. For randomized non-preemptive algorithms, we establish a tight $\Theta(\sqrt{n/m})$ competitive ratio, which also improves the best offline approximation to $O(\sqrt{n/m})$. For deterministic non-preemptive algorithms on multiple machines, we prove an $O(n/m^2 + \sqrt{n/m}\log m)$ upper bound and an $\Omega(n/m^2 + \sqrt{n/m})$ lower bound. In the kill-and-restart model, we reveal a sharp transition for deterministic algorithms: $\Omega(n/\log n)$ for $m = 1$ versus $\Theta(\sqrt{n/m})$ for $m \ge 2$; the latter matches the optimal randomized ratio, and we further show that randomization provides no additional power in this model.

We also investigate the setting where $n$ is unknown. We prove that no randomized non-preemptive algorithm achieves $o(n)$ on one machine or $o(n/m^2 + \sqrt{n/m})$ on $m$ machines. In contrast, our kill-and-restart algorithm achieves $O(n^{\alpha}/\sqrt{m})$ for $m \ge 2$, where $\alpha = (\sqrt{5}-1)/2$, breaking the trivial bound without knowledge of $n$.
\end{abstract}



\section{Introduction}

Scheduling is one of the most fundamental problems in combinatorial optimization, with deep roots in both theory and practice. A rich landscape of scheduling variants spans different machine environments (identical, related, and unrelated~\cite{DBLP:journals/algorithmica/EpsteinS04,DBLP:journals/mp/LenstraST90,DBLP:journals/tcs/GairingMW07}) and optimization objectives (makespan~\cite{DBLP:journals/siamam/Graham69}, completion time~\cite{DBLP:journals/mor/HallSSW97}, and flow time~\cite{DBLP:conf/soda/Bansal04,DBLP:journals/siamcomp/KellererTW99,DBLP:journals/jcss/LeonardiR07}). A prominent line of research focuses on the \emph{online} setting~\cite{graham1966bounds,DBLP:books/daglib/0097013}, where jobs arrive over time and the scheduler must make irrevocable decisions without knowledge of future arrivals. In modern applications such as cloud resource management and data center scheduling~\cite{DBLP:journals/jnsm/JenningsS15, DBLP:conf/nsdi/GhodsiZHKSS10,DBLP:conf/analco/Fleischer10}, an online service platform must dispatch jobs submitted by multiple users as they arrive; this naturally leads to the objective of minimizing the total \emph{flow time} (or equivalently, the average flow time), defined as the sum over all jobs of the difference between completion time and release time.

In this paper, we focus on the flow time minimization problem on $m$ identical machines. Most positive results (e.g., \cite{DBLP:conf/esa/LucarelliTST16,DBLP:conf/stoc/ChadhaGKM09,DBLP:conf/soda/ChenIP25,DBLP:conf/icalp/GuptaKL18,DBLP:journals/dam/EpsteinS06}) on the flow time objective require additional capabilities --- such as preemption, resource augmentation, and rejection --- to surpass the trivial greedy algorithm\footnote{The trivial greedy algorithm simply schedules any waiting job whenever any machine becomes available, achieving a competitive ratio of $\Theta(n / m)$.}. This is due to the following single-machine lower bound, which holds even when $n$ is known in advance:
\begin{construction}
  The adversary first releases a unit job at time $0$. As soon as the algorithm starts processing it or after waiting for $n$ units of time, the adversary releases $n - 1$ $\varepsilon$-jobs.
\end{construction}
Any deterministic non-preemptive algorithm is thus forced to incur $\Omega(n)$ flow time, whereas the offline optimum remains $O(1)$~\cite{DBLP:journals/siamcomp/KellererTW99,DBLP:conf/stoc/ChekuriKZ01}.
However, this lower bound applies only to \emph{deterministic} algorithms in the \emph{single-machine} setting. The limitations of this lower bound motivate our exploration of the following open questions.

\paragraph{Utilizing Randomness.} \emph{Randomization} is widely employed in the design of online algorithms. Against an oblivious adversary who does not adapt to the randomness used by the algorithm, randomization often breaks deterministic barriers, including the classic caching problem and matching problem~\cite{DBLP:journals/jal/FiatKLMSY91, DBLP:conf/stoc/KarpVV90}.

In our problem, the previously known lower bound $\Omega(\sqrt{n})$ for oblivious adversaries by \cite{DBLP:journals/tcs/EpsteinS03} applies only to the single-machine case; for $m \geq 2$, no information-theoretic lower bound stronger than $\Omega(1)$ was previously known, leaving a substantial gap from the $O(n/m)$ competitive ratio of the trivial greedy algorithm.

\begin{openq}{Open Question 1: Randomized Algorithms}
 Can we use randomization to beat the trivial greedy algorithm?\\
 Is it possible to find a randomized algorithm much better than $\Theta(\sqrt{n})$ (perhaps even a constant) when $m \geq 2$?
\end{openq}

\paragraph{Deterministic Multi-Machine Case.}
In fact, even without randomization, the problem remains unsolved when $m$ is large: the deterministic lower bound $\Omega(n/m^2)$ established by \cite{DBLP:journals/dam/EpsteinS06} decreases much faster than the trivial upper bound $\Theta(n/m)$ achieved by the greedy algorithm, leaving a substantial gap when $m$ is comparable to $n$.

\begin{openq}{Open Question 2: Deterministic Algorithms}
Is greedy indeed the best deterministic algorithm for the multi-machine case? \\
What is the optimal competitive ratio when $m$ is large?
\end{openq}

\paragraph{Kill-and-Restart Model.} A closely related direction is the
\emph{kill-and-restart} model, first introduced by~\cite{DBLP:journals/siamcomp/ShmoysWW95}, where the scheduler may kill a running job and restart it from scratch, losing all processed work. This model is particularly well-suited to applications where partially processed work cannot be resumed (e.g., data corruption or non-checkpointable tasks). Unlike preemption, it avoids the need to maintain and store potentially unbounded intermediate states of suspended jobs, preserving the non-preemptive nature of job execution while giving the scheduler additional flexibility to correct past decisions.

Kill-and-restart has proven powerful for other scheduling objectives. For makespan and total completion time, it surpasses known non-preemptive barriers~\cite{DBLP:conf/approx/0002KTWZ18,DBLP:journals/jal/SteeP05}; most notably, \cite{DBLP:journals/mp/JagerSWW25} showed that for non-clairvoyant scheduling under total completion time, kill-and-restart improves the competitive ratio from $\Omega(n)$ to $O(1)$.

Despite these successes, the potential of kill-and-restart for flow time minimization has remained largely unexplored. No algorithm was previously known to exploit this capability, and the only known lower bound $\Omega(\sqrt{n})$~\cite{DBLP:journals/tcs/EpsteinS03} is restricted to the single-machine case. It thus remained unclear whether, and to what extent, kill-and-restart can improve upon the $O(n/m)$-competitive greedy algorithm.

\begin{openq}{Open Question 3: Kill-and-Restart}
Can we use kill-and-restart capability to beat the trivial greedy algorithm?\\
Is it possible to find an algorithm under the kill-and-restart model that is much better than $\Theta(\sqrt{n})$ (perhaps even a constant) when $m \geq 2$?
\end{openq}

\subsection{Our Results with Known \texorpdfstring{$n$}{n}}

We give a complete set of answers to Open Questions 1--3 when $n$ is known in advance, summarized in Table~\ref{tab:res_single} (the single-machine case) and Table~\ref{tab:res_multi} (the multi-machine case). A key technical feature of our lower bounds in the multi-machine case is their robustness: unlike bounds that only apply to specific choices of these parameters, they hold for every fixed pair of $(n, m)$.

Throughout this paper, we assume $n \ge m$ when presenting competitive ratio; when $n < m$, all jobs can be processed immediately without waiting, yielding optimal schedules with a competitive ratio of $1$.

\subsubsection*{Randomized Algorithms: Beating the Deterministic Barrier}
We show that randomization can overcome the $\Omega(n)$ deterministic barrier. For the single-machine case, we present an $O(\sqrt{n})$-competitive randomized algorithm, matching the lower bound of \cite{DBLP:journals/tcs/EpsteinS03}. For the multi-machine case, we achieve an $O(\sqrt{n/m})$ competitive ratio and prove a matching $\Omega(\sqrt{n/m})$ lower bound that holds for every pair $(n, m)$. Together, these results completely characterize the power of randomization in non-preemptive flow time minimization.

\subsubsection*{Deterministic Algorithms: Closing the Multi-Machine Gap}
The single-machine case is already resolved, with an optimal competitive ratio of $\Theta(n)$. In the multi-machine case, we nearly close the gap between the known $\Omega(n/m^2)$ lower bound and the trivial $O(n/m)$ upper bound, pinning down the optimal ratio to $\tilde{\Theta}(n/m^2 + \sqrt{n / m})$.\footnote{$\tilde{\Theta}(\cdot)$ hides polylogarithmic factors.}

\subsubsection*{The Power of Kill-and-Restart}

Perhaps our most surprising finding is the non-uniform impact of the kill-and-restart capability on deterministic algorithms.

\begin{itemize}
    \item For the single-machine case, kill-and-restart offers little advantage: we prove a lower bound of $\Omega(n / \log n)$.
    \item For $m \geq 2$, the competitive ratio drops dramatically. We present a deterministic kill-and-restart algorithm achieving $O(\sqrt{n/m})$, a quadratic improvement over the single-machine case.
\end{itemize}

Remarkably, this deterministic kill-and-restart ratio of $O(\sqrt{n/m})$ matches the optimal randomized bound. We further show that this coincidence is tight: the $\Omega(\sqrt{n/m})$ lower bound holds even for randomized algorithms with kill-and-restart.

\subsubsection*{Improving the State-of-the-Art for Offline Scheduling}
A direct consequence of our online results is a new state-of-the-art for the classic offline problem. We improve the long-standing best approximation ratio from $O(\sqrt{n/m} \log(n/m))$ by \cite{DBLP:journals/jcss/LeonardiR07} to $O(\sqrt{n/m})$. The improved bound can be realized by either derandomizing our randomized algorithm or applying our deterministic kill-and-restart algorithm in an offline context where restarts can be ignored at no cost.

\begin{table}[t]
  \belowrulesep=0pt
  \aboverulesep=0pt
  \centering
  \caption{Summary of competitive ratios for the single-machine setting.}
  \vspace{-0.8em}
  \caption*{\small \colorbox{FixedBlue!10}{Shaded cells} indicate results from this paper.}
  \vspace{-0.3em}
  \renewcommand{\arraystretch}{1.5}
  \label{tab:res_single}
  \begin{tabular}{c|c|c|c}
    \toprule
      \multicolumn{2}{c|}{Single machine, $n$ jobs} & \textbf{Non-Preemptive} & \textbf{Kill-and-Restart} \\
    \midrule
    \multirow{4}{*}{\textbf{Deterministic}} & \multirow{2}{*}{\textsc{Upper Bound}} & $O(n)$ & $O(n)$ \\
        & & (Trivial) & (Trivial) \\
    \hhline{~---}
    & \multirow{2}{*}{\textsc{Lower Bound}} & $\Omega(n)$ & \cellcolor{FixedBlue!10} $\Omega(n/\log n)$ \\
    & & \cite{DBLP:journals/siamcomp/KellererTW99,DBLP:conf/stoc/ChekuriKZ01} & \cellcolor{FixedBlue!10} \Cref{thm:restart-lb} \\
    \hhline{----}
    \multirow{4}{*}{\textbf{Randomized}} & \multirow{2}{*}{\textsc{Upper Bound}} & \cellcolor{FixedBlue!10} $O(\sqrt{n})$ & \cellcolor{FixedBlue!10} $O(\sqrt{n})$ \\
    & & \cellcolor{FixedBlue!10} \Cref{thm:random-no-restart-alg-single} & \cellcolor{FixedBlue!10} \Cref{thm:random-no-restart-alg-single} \\
    \hhline{~---}
    & \multirow{2}{*}{\textsc{Lower Bound}} & $\Omega(\sqrt{n})$ & $\Omega(\sqrt{n})$ \\
    & & \cite{DBLP:journals/tcs/EpsteinS03} & \cite{DBLP:journals/tcs/EpsteinS03} \\
    \bottomrule
  \end{tabular}
\end{table}

\begin{table}[t]
  \belowrulesep=0pt
  \aboverulesep=0pt
  \centering
  \caption{Summary of competitive ratios for the multi-machine setting.}
  \vspace{-0.8em}
  \caption*{\small \colorbox{FixedBlue!10}{Shaded cells} indicate results from this paper.}
  \vspace{-0.3em}
  \renewcommand{\arraystretch}{1.5}
  \label{tab:res_multi}
  \begin{tabular}{c|c|c|c}
    \toprule
    \multicolumn{2}{c|}{{$m$ machines, $n$ jobs}} & \textbf{Non-Preemptive} & \textbf{Kill-and-Restart} \\
    \midrule
    \multirow{4}{*}{\textbf{Deterministic}} & \multirow{2}{*}{\textsc{Upper Bound}} & \cellcolor{FixedBlue!10} $O(n/m^2 + \sqrt{n/m}\log m)$ & \cellcolor{FixedBlue!10}
$O(\sqrt{n/m})$  for ${m \geq 2}$ \\
    & & \cellcolor{FixedBlue!10} \Cref{thm:det-no-restart} & \cellcolor{FixedBlue!10} \Cref{thm:det_restart} \\
    \hhline{~---}
    & \multirow{2}{*}{\textsc{Lower Bound}} & \cellcolor{FixedBlue!10} $\Omega(n/m^2 + \sqrt{n/m})$ & \cellcolor{FixedBlue!10} $\Omega(\sqrt{n/m})$ \\
    & & \cellcolor{FixedBlue!10} \Cref{thm:no-restart-lb} & \cellcolor{FixedBlue!10} \Cref{thm:multi-lb-restart} \\
    \hhline{----}
    \multirow{4}{*}{\textbf{Randomized}} & \multirow{2}{*}{\textsc{Upper Bound}} & \cellcolor{FixedBlue!10} $O(\sqrt{n/m})$ & \cellcolor{FixedBlue!10} $O(\sqrt{n/m})$
\\
    & & \cellcolor{FixedBlue!10} \Cref{thm:random-no-restart-multi-alg} & \cellcolor{FixedBlue!10} \Cref{thm:random-no-restart-multi-alg} \\
    \hhline{~---}
    & \multirow{2}{*}{\textsc{Lower Bound}} & \cellcolor{FixedBlue!10} $\Omega(\sqrt{n/m})$ & \cellcolor{FixedBlue!10} $\Omega(\sqrt{n/m})$ \\
    & & \cellcolor{FixedBlue!10} \Cref{thm:multi-lb} & \cellcolor{FixedBlue!10} \Cref{thm:multi-lb-restart} \\
    \bottomrule
  \end{tabular}
\end{table}

\subsection{Our Results without Prior Knowledge of \texorpdfstring{$n$}{n}}
\label{subsec:unknown_n}

All results above assume that the total number of jobs $n$ is known in advance, a convention shared by prior positive results on non-preemptive flow time minimization~\cite{DBLP:journals/dam/EpsteinS06,DBLP:journals/tcs/EpsteinS03}.
This assumption is critical for deterministic algorithms: \cite{DBLP:journals/dam/EpsteinS06} showed that without knowledge of $n$, no deterministic algorithm on any constant number of machines can achieve a competitive ratio of $o(n)$, even when compared against a single-machine offline optimum. Is this assumption also critical for randomized algorithms and algorithms with kill-and-restart capability?

\begin{openq}{Open Question 4: Unknown Number of Jobs}
Can randomization or kill-and-restart still overcome the deterministic $\Theta(n)$ barrier when $n$ is not known in advance?
\end{openq}

We find a sharp separation between the two.

\paragraph{Randomization is insufficient.}
We prove that without knowledge of $n$, no randomized non-preemptive algorithm can achieve a competitive ratio of $o(n)$ on a single machine, or $o(n/m^2 + \sqrt{n/m})$ on $m$ machines.

\paragraph{Kill-and-restart overcomes the barrier.}
In contrast, we extend our deterministic kill-and-restart algorithm to operate without knowledge of $n$, achieving a competitive ratio of $O(n^{\alpha}/\sqrt{m})$ for all $m \ge 2$, where $\alpha = (\sqrt{5}-1)/2 \approx 0.618$.
This significantly improves over the trivial $O(n/m)$ bound and demonstrates that kill-and-restart provides a fundamental advantage that persists even without knowledge of $n$.

\paragraph{Algorithms with predicted $n$.} Finally, we note that if an approximate prediction $\hat{n} \in [n/c, cn]$ is available instead of the exact value of $n$, all our known-$n$ algorithms' guarantees degrade by at most a constant factor of $c$, preserving their asymptotic performance.

\subsection{Our Techniques}
\label{subsec:tech}

A central challenge in non-preemptive scheduling is to partition incoming jobs into two manageable groups: a small set of \emph{large} jobs, and a set of \emph{small} jobs whose processing times are provably limited.  
This distinction is key to circumventing the classic $\Omega(n)$ lower bound, where handling a single large job at the wrong time can delay a swarm of subsequent small jobs, leading to high flow time. Offline algorithms can address this by partitioning jobs based on global properties of the total job set (e.g., the offline approximation algorithms by~\cite{DBLP:journals/siamcomp/KellererTW99} and \cite{DBLP:journals/jcss/LeonardiR07} relied on a pre-solved preemptive solution over all jobs, and used the flow time of jobs in the preemptive solution to partition jobs), but this is impossible to implement in an online context, since future jobs are unknown. Therefore, a fundamentally new approach is required for the online setting where decisions are irrevocable and based on incomplete information.

To overcome this barrier, we introduce a new online rank-based partitioning framework, which is the cornerstone of all our algorithms. Instead of using an unknown future metric (like the flow time in the preemptive solution over all jobs) to classify a job's size, we use its processing time $p_j$, an immediately available property. Jobs are dynamically classified based on their rank: at any time, we designate the $k$ jobs with the largest processing times seen so far as \emph{large}, with all others classified as \emph{small}. This rank-based distinction allows the algorithm to adapt as new jobs arrive. However, it also introduces two significant technical challenges that we resolve:

\begin{itemize}
    \item[\textbf{(1)}] \textbf{Managing dynamic re-classification.} A job initially classified as large may later become small as larger jobs arrive. If we start treating it as a small job from that point on, we must carefully bound its previous waiting time and analyze whether this transition delays future jobs by effectively leaving one more job behind \OPT. 
    To handle this transition gracefully, we introduce the concept of \emph{proxy jobs}. Instead of relabeling a large job when it should become small, we treat it as a newly arrived small job, called a \emph{proxy} job. We then decompose its total flow time into two parts: the waiting time during its large phase, and the flow time of the proxy job. Note that a proxy job is not an actual job from the original instance, so we must further bound the additional impact introduced by creating such proxy jobs.

    \item[\textbf{(2)}] \textbf{Scheduling small jobs without a global view.} Prior offline approaches use a near-optimal preemptive schedule as a guide, and try to simulate the preemptive schedule of small jobs in a non-preemptive way. However, without the guide, we must handle the stream of small jobs using only local information. We demonstrate that a simple greedy strategy, Non-Preemptive Shortest Job First (\NSJF), is remarkably effective. A key part of our technical contribution is a novel and robust analysis of \NSJF. We analyze its performance under the influence of \emph{blocking periods}, where machines may be temporarily occupied by large jobs. This analysis allows us to precisely bound the flow time of small jobs, even with interference.
\end{itemize}

With this online rank-based partitioning framework in place, subsequent algorithmic design can focus solely on coordinating a bounded number of large jobs with an efficiently scheduled set of small jobs. A key advantage of this framework is that it no longer requires a preemptive algorithm to exploit global properties of the entire job set. This distinction not only enables our algorithm to be implemented online, but also improves the offline approximation ratio by a logarithmic factor over the result of \cite{DBLP:journals/jcss/LeonardiR07}, since their logarithmic loss arises from relying on a suboptimal preemptive algorithm in the multi-machine setting.

We remark that our online partitioning framework and the analysis of \NSJF do not inherently rely on the machines being identical, and we expect them to be useful for broader settings such as related machines.

\subsection{Paper Organization}

The remainder of this paper is organized as follows. In \Cref{sec:partition}, we introduce our core technical contribution: the online rank-based partitioning framework that serves as the foundation for all our algorithms, and the analysis of the \NSJF algorithm for handling small jobs within this framework. In the deterministic algorithms, we use a generalized analysis of the \NSJF algorithm, which is detailed in \Cref{sec:SJF}.

We then present our main algorithmic results. \Cref{sec:single-random} develops our $O(\sqrt{n/m})$-competitive randomized non-preemptive algorithm. \Cref{sec:multi-det-no-restart} presents our deterministic non-preemptive $O(n/m^2 + \sqrt{n/m}\log m)$-competitive algorithm. \Cref{sec:det-restart} details our deterministic $O(\sqrt{n/m})$-competitive algorithm using kill-and-restart for $m \ge 2$, and then we extend it to the unknown-$n$ setting in \Cref{sec:unknownn_det_restart}.

Finally, we establish the tightness of our results in \Cref{sec:lower} and \Cref{sec:unknownn_lower}, where we present a comprehensive set of matching lower bounds for both the known-$n$ and unknown-$n$ settings, respectively.

\subsection{Related Works}

The problem of minimizing total flow time has been a central and extensively studied topic in scheduling theory for decades, spanning a wide range of models and settings.

\paragraph{Non-Preemptive Scheduling.}
In the classic non-preemptive offline setting, the problem is known to be NP-hard. A key result by \cite{DBLP:journals/jcss/LeonardiR07} provided an $O(\sqrt{n/m} \log(n/m))$-approximation algorithm, which stood as the best-known guarantee for many years. On the approximation hardness front, \cite{DBLP:journals/siamcomp/KellererTW99} established a polynomial-time lower bound of $\Omega(\sqrt{n})$ for the single-machine case, later complemented by \cite{DBLP:journals/jcss/LeonardiR07}, who presented an $\Omega(n^{1/3 - \varepsilon})$ lower bound for multiple machines when $n \geq m ^ {4/\varepsilon}$. In the online setting, the problem is notoriously difficult, with a simple greedy algorithm achieving a $\Theta(n/m)$ competitive ratio, which for a long time was the best-known result for deterministic algorithms.

\paragraph{Preemptive Scheduling.}
To circumvent the hardness of non-preemptive scheduling, a significant body of work has focused on the preemptive model, where jobs can be suspended and resumed. In this setting, the folklore \SRPT (Shortest Remaining Processing Time) policy is optimal on a single machine. For multiple machines, \SRPT achieves a tight competitive ratio of $\Theta(\log(n/m))$ when job migration is allowed \cite{DBLP:journals/jcss/LeonardiR07}. For the more restrictive model without migration, \cite{DBLP:journals/siamcomp/AwerbuchALR02} presented an algorithm with a competitive ratio of $O(\log(\min\{\log n, \log P\}))$.

\paragraph{Models with Algorithmic Relaxations.}
Another prominent line of research has explored models that grant additional power to the online non-preemptive algorithm. One common approach is \emph{resource augmentation}, where the algorithm is provided with faster or more machines than the adversary. \cite{DBLP:journals/algorithmica/PhillipsSTW02} showed that an $O(1)$-competitive ratio is achievable with $m \log P$ machines, where $P$ is the max-to-min job size ratio. Another popular relaxation is the \emph{rejection} model, where the algorithm can choose to reject a small fraction of jobs.  \cite{DBLP:journals/jcss/ChoudhuryDGK18} first introduced this notion for flow time, with subsequent works \cite{DBLP:conf/icalp/GuptaKL18, DBLP:conf/esa/LucarelliMTST18} achieving an $O(1/\varepsilon^3)$-competitive algorithm for rejecting an $\varepsilon$-fraction of jobs.

\paragraph{The Kill-and-Restart Model.}
The kill-and-restart capability, introduced by \cite{DBLP:journals/siamcomp/ShmoysWW95} in the context of makespan minimization, represents a minimal form of preemption. This model has been successfully applied to other objectives, demonstrating its power to surpass non-preemptive lower bounds. For minimizing total completion time, \cite{DBLP:journals/jal/SteeP05} designed a $3/2$-competitive algorithm, breaking the non-preemptive barrier of $1.582$. Even more dramatically, for non-clairvoyant scheduling, \cite{DBLP:journals/mp/JagerSWW25} showed that kill-and-restart reduces the competitive ratio for total completion time from $\Omega(n)$ to $O(1)$. Despite these successes, the potential of kill-and-restart for flow time minimization has remained largely unexplored. The primary work in this specific area was by \cite{DBLP:journals/tcs/EpsteinS03}, who established several lower bounds, including an $\Omega(\sqrt{n})$ bound for randomized algorithms on a single machine, but no upper bounds were known.

\section{Preliminaries}
\label{sec:pre}

\subsection{Problem Formulation}

An instance of our scheduling problem is defined by a set of $m$ identical machines and a set of $n$ jobs, $J = \{1, 2, \dots, n\}$. Each job $j \in J$ is characterized by\footnote{In our algorithms, when comparing job release times or processing times, we break ties in an arbitrary but fixed manner. For convenience, one may assume that all values are distinct.}:
\begin{itemize}
    \item A \emph{release time} $r_j$, at which the job becomes available for processing.
    \item A \emph{processing time} $p_j$, representing the duration required to complete the job.
\end{itemize}
A job becomes known to the algorithm at its release time $r_j$ and must be assigned to a machine at or after $r_j$. Each job must be processed on a single machine at a time and requires $p_j$ units of processing to complete. Each machine can process at most one job at a time.

We introduce standard terminology for describing the state of jobs and machines. A machine is said to be \emph{busy} at time~$t$ if it is processing a job, and \emph{idle} otherwise. Time~$t$ is said to be \emph{busy} if all machines are busy at time~$t$. A job~$j$ is said to be \emph{waiting} at time~$t$ if it has been released but is not currently being processed by any machine.

\subsection{Scheduling Models}

This paper focuses on two primary models:

\begin{itemize}
    \item \textbf{Non-Preemptive:} Once a job begins execution on a machine, it must run to completion without interruption.
    \item \textbf{Kill-and-Restart:} The algorithm may terminate (kill) a job that is currently running. The job returns to the pool of available jobs and can be started again from the beginning at a later time, possibly on a different machine. All work done prior to the termination is lost.
\end{itemize}

For the purpose of analysis, we also refer to the following preemptive models.

\begin{itemize}
    \item \textbf{Preemptive (without Migration):} The algorithm may interrupt a job and later resume it on the same machine without losing any processed work.
    \item \textbf{Preemptive with Migration:} The algorithm may interrupt a job and later resume it on any machine without losing any processed work.
\end{itemize}

\subsection{Objective and Performance Metric}
We define $C_j$ as the time when job~$j$ finishes its last unit of processing and its flow time as $F_j = C_j - r_j$. For the non-preemptive and kill-and-restart models, we write $s_j$ for the start time of job~$j$'s (final successful) execution, and we have $C_j = s_j + p_j$. The objective is to minimize the total flow time:
$$
F(\S) = \sum_{j \in J} F_j.
$$

\paragraph{Competitive Analysis.} The performance of an online algorithm is measured by its competitive ratio. An algorithm $\ALG$ is $\Gamma$-competitive\footnote{The competitive ratio~$\Gamma$ is often a function $\Gamma(n, m)$ of the number of jobs $n$ and machines $m$ in this paper.} if for any job instance $J$:
$$ F(\ALG(J)) \leq \Gamma \cdot F(\OPT(J)), $$
where $\OPT(J)$ is the total flow time of an optimal offline solution. We define $\OPT$ differently for upper and lower bounds:
\begin{itemize}
    \item For upper bounds, $\OPT$ is the optimal offline \emph{preemptive} schedule \emph{without migration}.
    \item For lower bounds, $\OPT$ is the optimal offline \emph{non-preemptive} schedule.
\end{itemize}
For convenience, we use $\OPT$ and $\ALG$ to denote the costs $F(\OPT(J))$ and $F(\ALG(J))$, respectively, when the context is clear.

\section{Online Rank-Based Partitioning Framework}

\label{sec:partition}

This section introduces our core technical contribution: a framework for partitioning jobs into \emph{large} and \emph{small} categories in an online manner. Unlike offline methods that rely on global knowledge, our approach uses a fixed rank-based threshold, $\ell$, whose value will be set by our algorithms based on $n$ and $m$.

In \Cref{subsec:reclassification}, we introduce this framework, centered around the primary challenge of dynamic job re-classification and our solution using \emph{proxy} jobs. We formalize this logic in \Cref{alg:partition} and prove its key properties. In \Cref{subsec:nsjf}, we show why the resulting set of \emph{small} jobs is manageable and can be scheduled efficiently using a simple greedy algorithm.

\subsection{The Challenge of Dynamic Re-Classification}
\label{subsec:reclassification}

Our framework is built on a dynamic partitioning strategy designed to maintain two key invariants at all times:
\begin{itemize}
    \item The number of active jobs classified as \emph{large} is bounded by $\ell$.
    \item The processing time of any job classified as \emph{small} is bounded by $\OPT / \ell$.
\end{itemize}

To achieve this, we maintain a set of \emph{active large jobs}, denoted $\LargeSet$, containing up to $\ell$ jobs with the largest processing times seen so far. The logic of our framework defines a clear life cycle for every job based on its interaction with this set.

\paragraph{Arrival and Classification.}
Upon arrival, a job is evaluated. If it is not large enough to enter the top-$\ell$ group, it is permanently classified as a \emph{small job}. Otherwise, it is designated as an \emph{active large job} and added to the set $\LargeSet$.

\paragraph{Retirement and the Proxy Mechanism.}

The main complexity arises when a job $j \in \LargeSet$ is forced out (or \emph{retired}) at time $r_i$ due to the arrival of a new, larger job $i$. The state of $j$ at this moment determines its new classification:
\begin{itemize}
    \item If job $j$ is still waiting, it becomes \emph{proxied}. A new proxy job $j'$ with the same processing time of $j$ is created immediately, which is then handled along with the small jobs. This is the primary mechanism our analysis relies on.
    \item If job $j$ has already started processing by the online algorithm, it cannot be proxied, as the online algorithm must not revoke historical scheduling decisions. We call this a \emph{committed} large job. Although no longer in the active set $\LargeSet$, it must still be tracked and scheduled as a large job. This is a subtle point in the analysis and can be safely ignored for a high-level understanding, but it is an important case required for the correctness of our online algorithm.
\end{itemize}

\begin{algorithm}[ht]
\caption{Online Rank-Based Partitioning Step}
    \label{alg:partition}
    \SetKwFunction{FPartition}{Partition}
    \SetKwProg{Proc}{Procedure}{:}{end}
    \KwData{
        Active large job set $\LargeSet$, and every job's other current states
    }
    \Proc{\FPartition{$i$, $\ell$}}{
        \uIf{$|\LargeSet| < \ell$}{
            Add $i$ to $\LargeSet$\;
            Mark $i$ as \textit{large}\;
        }
        \uElseIf{$p_i > \min_{j \in \LargeSet} p_j$}{
            Let $j = \arg\min_{j \in \LargeSet} p_j$\tcp*{$j$ will be retired.}
            Remove $j$ from $\LargeSet$ and add $i$ to $\LargeSet$\;
            \uIf{job $j$ is still waiting}{
                Create a new \textit{proxy} job $j'$ with $r_{j'} = r_i$ and $p_{j'} = p_{j}$\;
                Mark $j'$ as \textit{small}\;
                Mark $j$ as \textit{proxied} by $j'$\;    \tcp{The scheduler will actually process $j$ when processing $j'$.}
            }
            \Else{
                Mark $j$ as \textit{committed large}\;
            }
            Mark $i$ as \textit{large}\;
        }
        \Else{
            Mark $i$ as \textit{small}\;
        }
    }
\end{algorithm}

We now present the pseudocode for this partitioning step in \Cref{alg:partition}. For simplicity, the function only explicitly manages the active large job set $\LargeSet$; the other sets (small, proxied, and committed jobs) are implicitly defined by the procedure. This partitioning mechanism guarantees our desired invariants along with a key monotonicity property: the set of all small jobs (including proxy jobs) forms an ``easier'' scheduling instance than the original job set. This is formalized in the following lemma.

\begin{lemma}
    \label{lem:num_partition}
    Under the online rank-based partitioning method, the following properties hold:
    \begin{enumerate}
        \item The number of active large jobs (the set $\LargeSet$) is at most $\ell$ at any time.
        \item The processing time of any job classified as small is upper bounded by $\OPT / \ell$.
        \item Let ${S}$ be the set of all jobs ever classified as small, including all proxy jobs. Let $J$ be the original job set. The instance ${S}$ is no harder to schedule than $J$, meaning $F(\OPT({S})) \leq F(\OPT(J))$.
    \end{enumerate}
\end{lemma}
\begin{proof}
\begin{itemize}
    \item[(1)] Holds by construction, as the algorithm explicitly maintains $|\LargeSet| \le \ell$.
    \item[(2)] Consider any job $j$ classified as small. By definition of the algorithm, at the time of its arrival, there must have been at least $\ell$ other jobs released that were larger than $j$. The total processing time of these $\ell$ larger jobs alone is a lower bound on the optimal solution's value, thus $p_j < \OPT / \ell$.
    \item[(3)] This property holds due to the proxy mechanism. The set ${S}$ differs from $J$ by replacing some proxied large jobs with proxy jobs. Crucially, each proxy job $j'$ is created at the exact moment $r_{j'}$ that a new, larger job $i \in J$ ($p_i > p_{j'}$) arrives and displaces it. We can therefore charge the existence of each proxy job in ${S}$ to a unique, strictly larger job in $J$ with the same release time. This implies that the instance ${S}$ is a ``weaker'' instance than $J$, and thus its optimal flow time can be no larger.
\end{itemize}
\end{proof}

\subsection{Handling Small Jobs with Non-Preemptive Shortest Job First}
\label{subsec:nsjf}

The partitioning framework effectively isolates a small, bounded number of active large jobs at any time, allowing our main algorithms to focus on managing them. This leaves the much larger set of small jobs, which must be handled efficiently. In this section, we show that a natural greedy algorithm, Non-Preemptive Shortest Job First (\NSJF), is provably effective for this task.

The \NSJF algorithm is simple: whenever a machine is idle and there are waiting jobs, it immediately starts processing the one with the smallest processing time. The following lemma guarantees the performance of \NSJF on a set of jobs whose sizes are bounded. This result is central to our analysis, as our framework guarantees that all jobs in the set ${S}$ have bounded size.

\begin{lemma}
\label{lem:sjfmain}
Consider running \NSJF on a set of jobs whose processing times are all bounded by $\tau$. Compared to an optimal offline preemptive schedule (even with migration), the total flow time of \NSJF is bounded by: 
    $F(\NSJF) \leq \OPT + O(n\tau)$.
\end{lemma}

\Cref{lem:sjfmain} demonstrates the power of this partition by bounding the flow time contribution from this large group of small jobs. For example, when we set the rank threshold $\ell = \sqrt{n}$, our framework ensures that all small jobs have a processing time $\tau \le \OPT/\sqrt{n}$. For clarity, we will first prove the lemma as stated. In \Cref{sec:SJF}, we prove a more general version of this result (\Cref{lem:SJFgeneral}) that is robust to interference from large jobs, where machines may be subject to initial ``blocking periods,'' during which they are unavailable for scheduling small jobs. Specifically, given a vector of initial blocking times $\vec{b}$ such that \NSJF cannot use a machine $i$ in its blocking period $[0,b_i)$, we show that \NSJF remains competitive with respect to the optimal flow time on an instance without blocking, up to an extra additive term depending on the total blocking length.
This general version is a fundamental component in the final analysis of our main algorithms.

\paragraph{Proof Strategy.}
The proof of \Cref{lem:sjfmain} follows a three-step structure. First, we show that the total volume of work \emph{started} by \NSJF quickly catches up to the volume of work \emph{processed} by the optimal solution after a small time shift. Second, we lift this volume-based argument to show that the \emph{number of jobs completed} by \NSJF also catches up to \OPT after a slightly larger time shift. Finally, we use this bound on the number of completed jobs to prove the lemma. Note that the general version of \Cref{lem:SJFgeneral} follows the same strategy; however, the analysis must additionally account for how the available processing power varies over time in the presence of blocking periods.

For any processing time threshold $p$ and time $t$, let $J_{\leq p}(t)$ be the set of jobs with $p_j \leq p$ and release time $r_j \leq t$.

\paragraph{Step 1: Volume Matching.}
For each algorithm, we track the volume of work associated with jobs in $J_{\leq p}(t)$. Let $V_{\leq p}(t)$ denote the total volume of work processed on jobs in $J_{\leq p}(t)$ by time $t$ under \OPT, and let $V'_{\leq p}(t)$ denote the total volume of jobs in $J_{\leq p}(t)$ that have been started by time $t$ under \NSJF.

\begin{lemma}[Volume Matching]
\label{lem:sjfvolume_warm}
For every $t \geq 0$ and every threshold $p$, we have
$$V'_{\leq p}(t + \tau) \geq V_{\leq p}(t).$$
\end{lemma}

\begin{proof}
Fix an arbitrary $p$, and let $t_0$ be the last time some machine is idle under $\NSJF$ before $t + \tau$. Let $j$ be the job with $p_j > p$ that is started latest by \NSJF within the interval $(t_0, t+\tau)$. If no such job exists, we let $s_j = t_0$ for a fictitious job $j$ of infinite size.

Since \NSJF starts job $j$ with $p_j > p$ at time $s_j$, it must have already started all jobs with size at most $p$ released up to $s_j$ (due to its SJF policy). Thus, we have:
$$
    V'_{\leq p}(s_j) = \sum_{h \in J_{\leq p}(s_j)} p_h \geq V_{\leq p}(s_j).
$$

We now analyze two cases. If $s_j\geq t$, we immediately have:
$$
    V'_{\leq p}(t + \tau) \geq V'_{\leq p}(s_j) \geq V_{\leq p}(s_j) \geq V_{\leq p}(t).
$$
Otherwise, if $s_j< t$, then since $s_j$ is the latest start of a job with $p_j > p$ after the last idle time $t_0$, all processing in $[s_j + \tau, t + \tau)$ is on jobs of size at most $p$, and no machine is idle in this interval. Therefore:
$$
    V'_{\leq p}(t + \tau) \geq V_{\leq p}(s_j) + m(t - s_j)  \geq V_{\leq p}(t).
$$
\end{proof}

\paragraph{Step 2: Completion Matching.}
Let $N_{\leq p}(t)$ and $N'_{\leq p}(t)$ denote the number of jobs in $J_{\leq p}(t)$ completed by time $t$ under \OPT and \NSJF, respectively. When the subscript ${\leq p}$ is omitted, we count all completed jobs regardless of size.

\begin{lemma}[Completion Matching]
\label{lem:sjfnumber_warm}
    For every $t \geq 0$ and every threshold $p$, we have
    $$N'_{\leq p}(t + 2\tau) \geq N_{\leq p}(t).$$
\end{lemma}
\begin{proof}
For each rank $i$, let $p_i$ be the size of the $i$-th smallest job completed under \OPT by time $t$, and let $p'_i$ be the size of the $i$-th smallest job started under \NSJF by time $t + \tau$. Let $k$ and $k'$ denote the total number of such jobs, respectively. We claim that $k' \ge k$ and $p'_i \le p_i$ for all $i \le k$.

Assume for contradiction that $i \le k$ is the smallest rank such that either $i > k'$ or $p'_i > p_i$. For \OPT, at least $i$ jobs of size at most $p_i$ are completed by time $t$, so:
$$
    V_{\leq p_i}(t) \geq \sum_{j=1}^{i} p_j.
$$
For \NSJF, all started jobs of size $\leq p_i$ are among the first $i - 1$ jobs (either because only $i - 1$ jobs have started, or because the $i$-th smallest is strictly larger than $p_i$). Therefore:
$$
    V'_{\leq p_i}(t + \tau) \leq \sum_{j=1}^{i-1} p'_j \leq \sum_{j=1}^{i-1} p_j,
$$
where the second inequality uses $p'_j \leq p_j$ for all $j < i$. This gives $V'_{\leq p_i}(t + \tau) < V_{\leq p_i}(t)$, contradicting \Cref{lem:sjfvolume_warm}.

Therefore $k' \ge k$ and $p'_i \leq p_i$ for all $i \le k$, which implies that for every threshold $p$, the number of jobs of size $\leq p$ started by \NSJF by time $t + \tau$ is at least the number completed by \OPT by time $t$. Since all jobs started by time $t + \tau$ are completed by time $t + 2\tau$, the lemma follows.
\end{proof}

\paragraph{Step 3: Bounding the Flow Time.}
\begin{proof}[Proof of \cref{lem:sjfmain}]
Let $u$ denote the maximum completion time under \NSJF. The total flow times can be expressed as:
$$
    F(\OPT) = \int_0^\infty \left( |J(t)| - N(t) \right) \, dt, \quad
    F(\NSJF) = \int_0^u \left( |J(t)| - N'(t) \right) \, dt,
$$
where $|J(t)|$ is the number of jobs released by time $t$.
\begin{align*}
    F(\NSJF)
    &= \int_0^u \left( |J(t)| - N'(t) \right) \, dt \\
    &\leq \int_0^{u - 2\tau} \left( |J(t)| - N'(t + 2\tau) \right) \, dt + \int_{u - 2\tau}^u |J(t)| \, dt \\
    &\leq \int_0^{u - 2\tau} \left( |J(t)| - N(t) \right) \, dt + 2n\tau \\
    &\leq F(\OPT) + 2n\tau.
\end{align*}
\end{proof}

\section{Randomized Non-Preemptive Algorithm}
\label{sec:single-random}

This section presents our $O(\sqrt{n/m})$-competitive randomized algorithm. We begin by detailing the core ideas in the single-machine setting before sketching the straightforward extension to multiple machines.

\subsection{Single-Machine Case}

Our algorithm is built on the online partitioning framework (\Cref{sec:partition}) with the large-job threshold set to $\ell = \lfloor \sqrt{n} \rfloor$. This ensures that all small jobs have a processing time of at most $\OPT/\sqrt{n}$. The high-level strategy is to combine two principles:
\begin{enumerate}
    \item \textbf{Prioritize Small Jobs:} A simple greedy strategy, \NSJF, is used to schedule all small jobs, ensuring their contribution to the total flow time is bounded by $O(\sqrt{n}) \cdot \OPT$ (\Cref{lem:sjfmain}).
    \item \textbf{Randomly Delay Large Jobs:} Each large job is carefully inserted into the schedule created by the small jobs. To minimize the disruption, each large job $j$ is randomly assigned a ``patience'' level, delaying its start until a certain amount of machine idle time has passed.
\end{enumerate}

Instead of presenting the online algorithm directly, we describe a \emph{dynamic algorithm} that reconstructs a schedule upon each job's arrival while keeping the decisions before unchanged. This dynamic view is a powerful tool for analysis, as it allows us to cleanly decompose the total flow time.

\begin{algorithm}[ht]
    \caption{Dynamic (Online Stable) Non-Preemptive Randomized Algorithm for Single-Machine}
    \label{alg:dynamic}
    \SetKwFunction{FPartition}{Partition}
    \SetKwInOut{Input}{Include} 
    \Input{\Cref{alg:partition}}
    \KwData{Every job's current states as determined by \Cref{alg:partition}.}
    \SetKwBlock{OnJobRelease}{\textbf{On job $i$ release:}}{}
    \SetKwBlock{OnEvent}{\textbf{On any job completion or after any job release:}}{}
    \OnJobRelease
    {
        Call \FPartition{$i, \ell$} with $\ell = \lfloor \sqrt{n} \rfloor$ to classify job $i$ and update related jobs\;
        \If{$i$ is classified as large} {
            Sample $w_i \sim \text{Unif}\{1, \dotsc, \lfloor \sqrt{n} \rfloor \}$
        }
        Run $\NSJF$ on all small jobs to get schedule $\S_1$ \tcp*{Including proxy jobs.}
        $\S_2 \gets \S_1$ \;
        \For{each unproxied (active or committed) large job $j$ in  order of release time $r_j$ } {
            $t \gets $ the first time in $\S_2$ where the cumulative idle time since $r_j$ is at least $w_j p_j$\;
            Insert $j$ into $\S_2$ to start at $t$ \tcp*{Overlapping jobs are shifted right.}
        }
        Schedule jobs according to $\S_2$\;
    }
\end{algorithm}

The key requirement, \emph{online stability}, is that the schedule reconstructed at time~$t$ must remain consistent with the scheduling decisions made prior to~$t$. The proof of this property is given below.

\begin{lemma}
\label{lem:single-online-simulatable}
\Cref{alg:dynamic} is online stable; that is, for every time $t$, its schedule before time $t$ depends solely on the jobs released before $t$.
\end{lemma}
\begin{proof}
Consider the arrival of a new job $i$ at time $r_i$, which triggers a reconstruction of the schedule. Let $\S'_1$ and $\S'_2$ be the schedules produced in the previous round (before $i$'s arrival), and let $\S_1$ and $\S_2$ be the newly computed schedules. We must show that for any time $t < r_i$, the schedules are identical.

First, consider the base schedule $\S_1$. The set of small jobs is only augmented by jobs with release time $r_i$ (either $i$ itself, if it is small, or a new proxy job). Since the set of small and proxy jobs released before $r_i$ is unchanged, and \NSJF is an online algorithm, $\S_1$ is identical to $\S'_1$ on the interval $[0, r_i)$.

Next, consider the final schedule $\S_2$. The set of unproxied large jobs changes only if job $i$ is large and retires a waiting large job $j$. In this case, $j$ is part of the \textsf{for}-loop that constructs $\S'_2$, but it is absent from the loop that constructs $\S_2$. However, because $j$ was retired while waiting at time $r_i$, its scheduled start time $s'_j$ in $\S'_2$ must be greater than or equal to $r_i$. The removal of $j$ from the insertion process for $\S_2$ can only affect the schedule at or after time $s'_j$. Any jobs scheduled in $\S'_2$ before time $s'_j$ will be scheduled at the same time in $\S_2$. Since $s'_j \ge r_i$, the schedule $\S_2$ is identical to $\S'_2$ on the interval $[0, r_i)$.

In all other cases, the set of unproxied large jobs remains the same, so the construction is identical up to time $r_i$. Therefore, the schedule is online stable.

\end{proof}

\begin{theorem}
    \label{thm:random-no-restart-alg-single}
    \Cref{alg:dynamic} is an online polynomial-time randomized non-preemptive algorithm for total flow time minimization on a single machine that is $O(\sqrt{n})$-competitive against the preemptive offline solution.
\end{theorem}
\begin{proof}
Recall that $F(\S)$ denotes the total flow time of a schedule $\S$. The final schedule produced by \Cref{alg:dynamic} has total flow time $F(\S_2)$. Our algorithm constructs this final schedule incrementally by transitioning through intermediate schedules $\S_1 \rightarrow \S_2$. We analyze the additional flow time incurred during this transition. 
Let $J$ denote the real job set, which includes only actual jobs and excludes proxy jobs. Let $L$ and $S$ denote the sets of large jobs (including all active and retired large jobs) and small jobs (including proxy jobs), respectively. Note that a job belongs to $L$ or $S$ according to its classification, large or small, at its arrival time.

First, $\S_1$ is competitive with the optimal schedule for the job set $S$. By \Cref{lem:num_partition}, all jobs in $S$ have size at most $\OPT / \lfloor \sqrt{n} \rfloor$. Therefore, by \Cref{lem:sjfmain}, we have:
$$
F(\S_1) = O(\sqrt{n}) \cdot F(\OPT(S)).
$$
Moreover, since $S$ is weaker than $J$ by \Cref{lem:num_partition}, $F(\OPT(S)) \leq \OPT$, and thus:
\begin{equation}
    \label{eqn:random_small}
    F(\S_1) = O(\sqrt{n}) \cdot \OPT.
\end{equation}

It remains to analyze the additional flow time incurred during the transition from $\S_1$ to $\S_2$ due to the insertion of large jobs.

Consider an unproxied large job $j$ inserted at time $t$. This insertion can delay previously scheduled jobs (large or small) by at most $p_j$ time units. Let $\delta_j$ denote the total increase in flow time due to such delays, which we refer to as the \emph{external delay} of job $j$. If a large job is proxied, it incurs no external delay, and thus we define $\delta_j = 0$. In \Cref{lem:delay}, we prove that only $O(\sqrt{n})$ jobs are delayed in expectation by any unproxied large job, giving:
\begin{equation}
    \label{eqn:random_exdely}
    \E[\delta_j] = O(\sqrt{n}) \cdot p_j.
\end{equation}

We must also account for the \emph{self-delay} of each large job. If $j$ is proxied by a virtual job $j'$ released at $r_{j'}$, then the actual flow time of $j$ is $C_{j'} - r_j$. Since $C_{j'} - r_{j'}$ is already included in $F(\S_1)$ via the flow time of $j'$, the only additional term to consider is $r_{j'} - r_j$.

On the other hand, if $j$ is not proxied and is directly inserted into the schedule, we only need to consider its current flow time upon insertion. Any subsequent increases to this flow time are charged to the external delays ($\delta_{j'}$) caused by later jobs.

Let $\hat{F}_j$ denote the self-delay of $j$:
$$
\hat{F}_j =
\begin{cases}
r_{j'} - r_j, & \text{if } j \text{ is proxied by } j' \\
\hat{C}_j - r_j,   & \text{if } j \text{ is unproxied}
\end{cases},
$$
where $\hat{C}_j$ is the completion time of $j$ in $\S_2$ immediately after its insertion (in the last update round of our algorithm). In \Cref{lem:self-blow-up}, we show:
\begin{equation}
    \label{eqn:random_self_large}
    \sum_{j \in L} \hat{F}_j = O(\sqrt{n}) \cdot \OPT.
\end{equation}

In conclusion, assuming the correctness of \Cref{lem:delay} (which proves \eqref{eqn:random_exdely}), \Cref{lem:self-blow-up} (which proves \eqref{eqn:random_self_large}), and using \eqref{eqn:random_small}, we obtain:
$$
F(\S_2) \leq F(\S_1) + \sum_{j \in L} \delta_j + \sum_{j \in L} \hat{F}_j \leq O(\sqrt{n}) \cdot \OPT + O(\sqrt{n}) \cdot \sum_{j \in L} p_j.
$$
Since $\sum_{j \in L} p_j \leq \OPT$, it follows that:
$$
F(\S_2) = O(\sqrt{n}) \cdot \OPT.
$$
Thus, we conclude that our algorithm is $O(\sqrt{n})$-competitive.

\end{proof}

\begin{lemma}[External Delay]
    \label{lem:delay}
    For each large job $j$, the additional flow time arising from its delay of other previous jobs is bounded by $O(\sqrt{n})\cdot p_j$ in expectation, i.e., $\E(\delta_j) = O(\sqrt{n})\cdot p_j$.
\end{lemma}
\begin{proof}
    Let $D_k$ be the set of jobs delayed by large job $j$ conditioned on the randomness $w_j = k$. To show that $\E[|D_{w_j}|] = O(\sqrt{n})$, the key claim is that every job $i~(i \neq j)$ appears in at most one of $D_1, D_2, \dots, D_{\lfloor \sqrt{n} \rfloor}$. This is because when inserting job $j$, there is exactly $p_j$ idle time between any two adjacent insertion locations.
    This prevents a single job from being delayed by more than one possible insertion position of $j$. Therefore, we have:
$$
\E[|D_{w_j}|] = \sum_{k=1}^{\lfloor \sqrt{n} \rfloor} |D_k| \cdot \Pr[w_j = k] \leq n \cdot \frac{1}{\lfloor \sqrt{n} \rfloor} = O(\sqrt{n}).
$$

Thus, the expected number of jobs delayed by $j$ is $O(\sqrt{n})$, and since each of them is delayed by at most $p_j$, this implies the lemma.
\end{proof}

\begin{lemma}[Self-Delay]
    \label{lem:self-blow-up}
    For all large jobs, the total self-delay is bounded by $O(\sqrt{n})\cdot \OPT$, i.e., $\sum_{j \in L} \hat{F}_j = O(\sqrt{n}) \cdot \OPT$. 
\end{lemma}
\begin{proof}
For a fixed large job $j$, the self-delay is either $r_{j'} - r_j$ (if proxied by $j'$) or $\hat{C}_j - r_j$ (if unproxied). We decompose the self-delay into three sources: collecting idle time, passing through busy time, and the job's own processing time $p_j$.

First, consider the idle time collected. Since $w_j \leq \sqrt{n}$, the total idle time accumulated for each large job $j$ is at most $\sqrt{n} \cdot p_j$. Summing over all large jobs:
$$
\sum_{j \in L} \text{(idle contribution)} \leq \sqrt{n} \cdot \sum_{j \in L} p_j \leq \sqrt{n} \cdot \OPT.
$$

Next, we analyze the busy time passed through. For each large job $j$, this occurs during the interval $[r_j, r_{j'})$ (if proxied) or $[r_j, \hat{s}_j)$ (if unproxied, where $\hat{s}_j$ denotes the start time of $j$ immediately after its insertion). In both cases, the large job remains unproxied throughout the interval. By \Cref{lem:num_partition}, the number of active unproxied large jobs at any time $t$ is at most $\sqrt{n}$, and there is at most one committed unproxied large job at any time $t$. Since the total length of busy intervals is at most $\OPT$:
$$
\sum_{j \in L} \text{(busy-time contribution)} \leq (\sqrt{n}+1) \cdot \OPT.
$$

Finally, the processing time contribution is $\sum_{j \in L} p_j \leq \OPT$. Combining all three:
$$
\sum_{j \in L} \hat{F}_j = O(\sqrt{n}) \cdot \OPT.
$$
\end{proof}

\subsection{Generalization to Multi-Machine Case}

The analysis for the multi-machine setting directly parallels the single-machine case. The core proof structure --- bounding the flow time of small jobs, external delay, and self-delay --- remains identical.  The generalization simply requires adjusting the key parameters to account for the presence of $m$ machines. Specifically, the threshold for large jobs becomes $\sqrt{nm}$, and we sample an additional random variable, $m_j$, for each large job $j$ to denote its assigned machine. These adjustments alter each component of the analysis as follows:

\begin{itemize}
    \item \textbf{Small Job Flow Time:} The large-job threshold $\ell$ is now set to $\sqrt{nm}$, which implies that the size of each small job is upper bounded by $\OPT/\sqrt{nm}$. Therefore, the total flow time incurred by small jobs is bounded by $O(\sqrt{n/m}) \cdot \OPT$.
    
    \item \textbf{External Delay:} The external delay for any large job $j$ is bounded by $O(\sqrt{n/m}) \cdot p_j$. In the single-machine setting, each job has $\sqrt{n}$ insertion choices, resulting in an expected delay affecting $\sqrt{n}$ jobs. In the multi-machine setting, we allocate $\sqrt{n/m}$ choices to $w_j$ and $m$ choices to $m_j$, resulting in an expected delay affecting only $O(\sqrt{n/m})$ jobs.
    
    \item \textbf{Self-Delay (Idle Time):} The self-delay incurred by collecting idle time is at most $\sqrt{n/m} \cdot p_j$, due to the smaller upper bound on $w_j$.
    
    \item \textbf{Self-Delay (Busy Time):} The self-delay incurred by passing through busy periods is bounded by $O(\sqrt{n/m}) \cdot \OPT$. This is because the number of waiting large jobs increases to $\sqrt{nm}$, while the total length of busy time across all machines decreases to $\OPT/m$.
\end{itemize}

For completeness, we present the full algorithm and its analysis in \Cref{sec:multi-random}.

\section{Deterministic Non-Preemptive Algorithm}
\label{sec:multi-det-no-restart}

This section presents a deterministic non-preemptive algorithm with a competitive ratio of $O(n/m^2 + \sqrt{n/m} \log m)$. This nearly closes the gap between the trivial $O(n/m)$ upper bound and the $\Omega(n/m^2)$ lower bound of~\cite{DBLP:journals/dam/EpsteinS06}, particularly when $m$ is large relative to $n$.

The algorithm builds on the online rank-based partitioning framework (\Cref{alg:partition}) with $\ell = \lfloor \sqrt{nm} \rfloor$.
Based on this classification, machines are assigned different roles to handle large and small jobs separately, as described below.

The first machine, called the \LargeOnly machine, is dedicated exclusively to large jobs and schedules the smallest waiting large job if available. For the remaining \All machines $2 \dots m$, a dynamic threshold function
$$
\gamma(k) = \min \left\{\left\lfloor \frac{k}{\sqrt{n/m}} \right\rfloor, \lceil m / 2 \rceil \right\}.
$$
controls how many of them may process large jobs, where $k$ is the number of active large jobs currently waiting. When one of these \All machines becomes idle, it schedules the smallest waiting job regardless of its size if the number of \All machines currently processing large jobs is less than $\gamma(k)$; otherwise, it schedules small jobs only using \NSJF. Since there are at most $\lceil m / 2 \rceil$ \All machines running large jobs at any time, there are always $m - 1 - \lceil m / 2 \rceil$ machines exclusively dedicated to small jobs.

We use $J$ to denote the real job set, which includes only real jobs and excludes proxy jobs.
Let $L$ denote the set of large jobs, and let $S$ denote the set of small jobs, including proxy jobs. 
Recall that, for a proxied large job $j$ (proxied by $j'$), its flow time consists of two components:
$
F_j = C_j - r_j = (C_{j'} - r_{j'}) + (r_{j'} - r_j).
$
The first term $C_{j'} - r_{j'}$, the flow time of the proxy job $j'$, is already included in $\sum_{j \in S} F_j$. Therefore, we define the following as a modified flow time:  
$$
\hat{F}_j = 
\begin{cases}
r_{j'} - r_j & \text{if it is proxied by $j'$} \\
C_j - r_j & \text{if it is unproxied}
\end{cases}.
$$
We have: 
$
\sum_{j \in L} \hat{F}_j + \sum_{j \in S} F_j = \sum_{j \in J} F_j. 
$
Let us begin by bounding the total flow time of jobs in $L$.

\begin{lemma} [Large Jobs]
    \label{lem:det-no-restart-large}
    $\sum_{j\in L} \hat{F}_{j} = O(\sqrt{n/m}) \cdot  \OPT.$
\end{lemma}
\begin{proof}
    Consider the waiting time of a large job $j$ during the interval $[r_j, s_j)$ (or $[r_j, r_{j'})$ if proxied). At any time $t$ in this interval, $j$ is an active large job that has not been scheduled, either because all \All machines are busy, or because at least $\gamma(k(t))$ of them are already processing large jobs.

    Let $\theta(t)$ denote the number of busy \All machines at time $t$ (processing any job, large or small), and let $k(t)$ denote the number of active large jobs waiting at time $t$. Since at least $\gamma(k(t))$ of these machines are processing large jobs, we have by the scheduling rule:
    $$
    \theta(t) = m - 1 \quad \text{or} \quad \gamma(k(t)) \leq \theta(t).
    $$
    By the definition of $\gamma$, this gives $k(t) = O(\sqrt{n/m}) \cdot (\theta(t) + 1)$.

    Whenever $k(t) \geq 1$, the \LargeOnly machine must also be busy, so the total number of busy machines is at least $\theta(t) + 1$. Since each busy machine contributes to the total processing time:
    $$
    \int_{0}^{\infty} (\theta(t) + 1) \cdot \mathbb{I}_{k(t) \geq 1}(t) \, dt \leq \sum_{j \in J} p_j.
    $$

    Combining with the processing time of the jobs themselves:
    \begin{align*}
    \sum_{j \in L} \hat{F}_j
    &\leq \sum_{j \in L} p_j + \int_0^\infty k(t)\,dt \\
    &\leq \sum_{j \in L} p_j + O(\sqrt{n/m}) \cdot \int_{0}^{\infty} (\theta(t) + 1) \cdot \mathbb{I}_{k(t)\geq 1}(t) \,dt \\
    &\leq \sum_{j \in L} p_j + O(\sqrt{n/m}) \cdot \sum_{j \in J} p_j
    = O\left( \sqrt{\frac{n}{m}} \right) \cdot \OPT.
    \end{align*}
\end{proof}

Next, we analyze the contribution of small jobs. Since the \LargeOnly machine is reserved for large jobs, small jobs are effectively scheduled on $m - 1$ machines. To compare against the $m$-machine optimum, we first bound the cost of losing one machine.

\begin{lemma}
    \label{lem:opt-m-m-k}
    For every $0 \leq k \leq m-1$, the optimal preemptive flow time on $m-k$ machines satisfies
    $\OPT_{m-k} \leq \left(1 + \frac{nk}{m(m-k)}\right) \cdot \OPT_{m}$,
    where both optima are under the preemptive model \emph{without migration}.
\end{lemma}

\begin{proof}
Starting from an optimal preemptive schedule on $m$ machines, we construct a feasible schedule on $m-k$ machines. Sort the machines by total workload and let $J_{[k]}$ denote the jobs on the $k$ lightest machines. Their total processing time satisfies $\sum_{j \in J_{[k]}} p_j \leq k \sum_{j \in J} p_j / m$.

Each job $j \in J_{[k]}$ is reassigned to one of the remaining $m-k$ machines, inserted at its original start time. This delays existing jobs on the target machine by at most $p_j$. To minimize impact, each job is placed on the machine with the fewest currently assigned jobs, ensuring at most $n/(m-k)$ jobs per machine. The additional flow time from inserting $j$ is therefore at most $n p_j / (m-k)$.

Summing over all reassigned jobs:
$$
\OPT_{m-k} \leq \OPT_m + \frac{n \sum_{j \in J_{[k]}} p_j}{m-k} \leq \OPT_m + \frac{nk}{m(m-k)} \cdot \OPT_m.
$$
\end{proof}

Our key tool for analyzing small jobs is the generalized \NSJF analysis (\Cref{lem:SJFgeneral}, proved in \Cref{sec:SJF}), which extends \Cref{lem:sjfmain} to handle \emph{initial blocking periods}. If \NSJF runs on $m'$ machines where machine $i$ is unavailable during $[0, b_i)$, and all job sizes are bounded by $\tau$, then
$$
F(\NSJF(\vec{b})) \leq F(\OPT(\vec{0})) + 2n\tau + \frac{nB}{m'},
$$
where $B = \sum_i b_i$ is the total blocking time and $\OPT(\vec{0})$ is the optimal preemptive schedule on $m'$ machines without blocking.

To apply this, we partition the timeline into intervals $[t_{h-1}, t_h)$, where $t_h$ is the start time of the $h$-th large job on an \All machine (excluding the \LargeOnly machine) and $t_0 = 0$. Let $S_h$ denote the small jobs released during $[t_{h-1}, t_h)$. Since no \All machine starts a new large job during $[t_{h-1}, t_h)$, all $m-1$ \All machines are available for small jobs, and every job in $S_h$ is scheduled before $t_h$. Some machines may be initially blocked by jobs still in progress from the previous interval: machine $i$ is unavailable during $[t_{h-1}, t_{h-1} + b_{h,i})$, where $b_{h,i}$ is its remaining processing time at $t_{h-1}$. Shifting the time origin to $t_{h-1}$, the flow time of $S_h$ under our algorithm matches that of $\NSJF(S'_h, \vec{b_h})$, where $S'_h$ is the shifted instance and $\vec{b_h}$ is the blocking vector.

It remains to bound the total blocking time.

\begin{lemma}[Blocking Time]
For each interval $(t_{h-1}, t_h]$, the total initial blocking time satisfies:
$$
\sum_{i=1}^{m-1} b_i = O\left( \sqrt{\frac{m}{n}} \log m \right) \cdot \OPT.
$$
\end{lemma}

\begin{proof}
At time $t_{h-1}$, each \All machine is blocked by at most one job still in progress. We bound the total blocking separately for small and large blocking jobs.

\paragraph{Small blocking jobs.} Each small job has size at most $O(\OPT/\sqrt{nm})$ by \Cref{lem:num_partition}. Since at most $m-1$ machines can be blocked, the total blocking from small jobs is at most $O(\sqrt{m/n} \cdot \OPT)$.

\paragraph{Large blocking jobs.} Order the large jobs that are still running on \All machines at time $t_{h-1}$ by their start times, and let $p_i$ denote the size of the $i$-th such job. When the $i$-th job was started, at least $i - 1$ other \All machines were already processing large jobs, so the scheduling rule required $\gamma(k) \geq i$, where $k$ is the number of active large jobs waiting at that moment. By the definition of $\gamma$, this implies $k \geq i\sqrt{n/m}$. Since these $k$ waiting jobs are all at least as large as $p_i$ (otherwise they would have been scheduled first), we have $k \cdot p_i \leq \OPT$, giving:
$$
p_i \leq \frac{\OPT}{i\sqrt{n/m}} = \frac{1}{i} \cdot O\left( \sqrt{\frac{m}{n}} \right) \cdot \OPT .
$$
There are at most $\lceil m/2 \rceil$ such jobs, so the total blocking from large jobs is:
$$
\sum_{i=1}^{\lceil m/2 \rceil} p_i \leq O\left( \sqrt{\frac{m}{n}} \right)  \cdot \OPT \cdot \sum_{i=1}^{\lceil m/2 \rceil} \frac{1}{i} = O\left( \sqrt{\frac{m}{n}} \log m \right) \cdot \OPT.
$$

Combining both parts gives the claimed bound.
\end{proof}

With the blocking time bounded, \Cref{lem:SJFgeneral} applies directly. We now prove the main result.

\begin{theorem}
\label{thm:det-no-restart}
There exists an online polynomial-time deterministic non-preemptive algorithm for total flow time minimization that is $O(n/m^2 + \sqrt{n/m} \log m)$-competitive against the preemptive offline solution.
\end{theorem}
\begin{proof}
By \Cref{lem:det-no-restart-large}, the total modified flow time of large jobs is bounded by:
$$
\sum_{j \in L} \hat{F}_j = O\left( \sqrt{\frac{n}{m}} \right) \cdot \OPT.
$$

Let $\tau$ be the size upper bound for small jobs, i.e., $\tau = O(\OPT / \sqrt{nm})$. For each small-job segment $[t_{h-1}, t_h)$, our algorithm performs the same as $\NSJF(S'_h)$ with initial blocking vector $\vec{b_h}$. By \Cref{lem:SJFgeneral} on $m-1$ machines, let $B_h = \sum_i b_{h,i}$:
$$
\sum_{j \in S_h} F_j \leq F(\OPT_{m-1}(S_h)) + 2 |S_h| \tau + \frac{|S_h| B_h}{m-1}.
$$
Summing over all disjoint small-job segments:
$$
\sum_{j \in S} F_j \leq F(\OPT_{m-1}) + 2n\tau + \frac{n \cdot \max_h B_h}{m-1}.
$$
Using the bounds $\tau = O(\OPT/\sqrt{nm})$ and $\max_h B_h = O(\sqrt{m/n} \log m \cdot \OPT)$, we get:
$$
\sum_{j \in S} F_j \leq \OPT_{m-1} + O\left( \sqrt{\frac{n}{m}} \right) \cdot \OPT + O\left( \sqrt{\frac{n}{m}} \log m \right) \cdot \OPT.
$$

By \Cref{lem:opt-m-m-k}, which bounds the loss in flow time from using $m-1$ instead of $m$ machines, we have:
$$
\OPT_{m-1} \leq \OPT + O(n/m^2) \cdot \OPT.
$$
Combining all terms:
\begin{align*}
\sum_{j \in S} F_j 
&\leq \OPT + O(n/m^2) \cdot \OPT + O\left( \sqrt{\frac{n}{m}} \right) \cdot \OPT + O\left( \sqrt{\frac{n}{m}} \log m \right) \cdot \OPT \\
&= O\left( \frac{n}{m^2} + \sqrt{\frac{n}{m}} \log m \right) \cdot \OPT.
\end{align*}

Finally, adding the contributions from large and small jobs:
$$
\sum_{j \in J} F_j = \sum_{j \in L} \hat{F}_j + \sum_{j \in S} F_j = O\left( \frac{n}{m^2} + \sqrt{\frac{n}{m}} \log m \right) \cdot \OPT.
$$
\end{proof}

\section{Deterministic Algorithm with Kill-and-Restart}
\label{sec:det-restart}

To utilize the power of kill-and-restart without relying on randomization, we first introduce a refined online partitioning framework with the following modification (see \Cref{alg:partition_restart} for details, and note that the threshold will be set to $\sqrt{nm}$):

\paragraph{Modified partition rule.} Let $P$ be the total size of all previously released jobs (excluding proxy jobs), which serves as a lower bound on $\OPT$. We classify a new job $j$ as a small job if $p_j \leq 4 P/\ell$, and large otherwise. The other procedures remain identical to \Cref{alg:partition}.

This modification preserves the original properties of the framework. The additional size filter provides a lower bound on the processing time of large jobs, which will be used in the analysis of kill events (\Cref{claim:fake_busy}).

\begin{algorithm}[ht]
\caption{Refined Online Partitioning for Kill-and-Restart Models}
    \label{alg:partition_restart}
    \SetKwFunction{FPartition}{Partition}
    \SetKwProg{Proc}{Procedure}{:}{end}
    \KwData{
        Active large job set $\LargeSet$, and every job's other current states
    }
    \Proc{\FPartition{$i, \ell$}}{
        Let $P$ be the sum of processing times of all (original) jobs released so far\;
        \uIf{$p_i \leq 4 P / \ell$}{
            Mark $i$ as \textit{small}\;
        }
        \Else{
            \If{$|\LargeSet| = \ell$}{
                Let $j = \arg\min_{j \in \LargeSet} p_j$\tcp*{$j$ will be retired.}
                Remove $j$ from $\LargeSet$\;
                \uIf{job $j$ is still waiting}{
                    Create a new \textit{proxy} job $j'$ with $r_{j'} = r_i$ and $p_{j'} = p_{j}$\;
                    Mark $j'$ as \textit{small}\;
                    Mark $j$ as \textit{proxied} by $j'$\;    \tcp{The scheduler will actually process $j$ when processing $j'$.}
                }
                \Else{
                    Mark $j$ as \textit{committed large}\;
                }
            }
            Add $i$ to $\LargeSet$\;
            Mark $i$ as \textit{large}\;
        }
    }
\end{algorithm}

\begin{lemma}
\label{lem:partition_restart}
Under the online rank-based partitioning framework (refined for restart), the following properties hold:
\begin{enumerate}
    \item The number of active large jobs (the set $\LargeSet$) is at most $\ell$ at any time.
    \item The processing time of any job classified as small, or any retired large job, is upper bounded by $O(\OPT / \ell)$.
    \item Let $S$ be the set of all jobs ever classified as small, including all proxy jobs. Let $J$ be the original job set. The instance $S$ is no harder to schedule than $J$, meaning $F(\OPT(S)) \leq F(\OPT(J))$.
    \item The processing time of any large job $j$ is at least $4P(r_j)/\ell$, where $P(r_j)$ denotes the total processing time of all jobs released up to $r_j$, excluding proxy jobs.
\end{enumerate}
\end{lemma}

\begin{proof}
The only property that requires verification under the modified framework is the size bound for small jobs. By the new partition rule, we classify a job $j$ as small if $p_j \leq 4P(r_j)/\ell$. Since $P(r_j)$ is a valid lower bound on $\OPT$, it follows that $p_j \leq O\left(\OPT/\ell\right)$, 
as required. The remaining properties are preserved from the original online partitioning framework.
\end{proof}

After introducing the refined partitioning framework, we now describe the basic design of our scheduling policy, which divides the $m$ machines into two roles:

\begin{itemize}
    \item {\SmallOnly machines ($1$ to $\lfloor m / 2 \rfloor$):} These machines accept \emph{only} small jobs.
    \item {\All machines ($\lfloor m / 2 \rfloor + 1$ to $m$):} These machines accept both small and large jobs.
\end{itemize}

Our scheduling policy builds on the partitioning produced by \Cref{alg:partition_restart}. Each job is classified as either large or small; each large job is further classified as active or retired, and as proxied or unproxied.

\paragraph{Scheduling rule.}
When an \All machine becomes idle, it schedules the smallest waiting unproxied job (regardless of whether it is large or small). When a \SmallOnly machine becomes idle, it schedules the smallest waiting small job.

\paragraph{Kill-and-restart rule.}
We say that a newly released small job is \emph{blocked} if any active large job is currently being processed, and \emph{unblocked} otherwise. The algorithm maintains a counter $\phi$, which is incremented by one for each blocked small job and reset to zero whenever an active large job begins processing or is killed. When $\phi$ reaches $\lfloor\sqrt{nm}\rfloor$, the algorithm kills all currently processing active large jobs and resets $\phi$ to zero.

The full pseudocode is presented in \Cref{alg:multi-deterministic-restart-m}.
\begin{algorithm} [h]
    \caption{Deterministic Algorithm with Kill-and-Restart on $m$ Machines}
    \label{alg:multi-deterministic-restart-m}
    \SetKwInOut{Input}{Include}  
    \Input{\Cref{alg:partition_restart}}
    \KwData{Blocked small job counter $\phi$ initialized to $0$; every job's current states as determined  by \Cref{alg:partition_restart}.}
    \SetKwBlock{OnJobRelease}{\textbf{On job $i$ release:}}{}
    \SetKwBlock{OnEvent}{\textbf{On any job completion or after any job release:}}{}
    \OnJobRelease
    {
        Call \texttt{Partition}$(i, \ell)$ with $\ell = \lfloor \sqrt{nm} \rfloor$ to classify job $i$ and update related jobs\;
         $j \gets \begin{cases}
            i & \text{if $i$ is classified as small} \\
            j' & \text{if $i$ is classified as large, and leads to a new proxy job $j'$} \\
            \bot & \text{otherwise (no new small job)}
        \end{cases}$ \;
        \If{$j \neq \bot$, i.e., there exists a new small job $j$}{
        Call $j$
        $\begin{cases}
            \emph{blocked} & \text{there exists a processing active large job} \\
            \emph{unblocked} & \text{no active large job is processing}
        \end{cases}$ \;
        $\phi \gets \phi + 1$ if $j$ is \emph{blocked} \;
        \If{$\phi = \lfloor \sqrt{nm} \rfloor$}
        {
            Kill all processing active large jobs and $\phi \gets 0$\;
        }
        }
    }
    \OnEvent
        {  
                \For {each idle \SmallOnly machine $i$} {
                    \If{ there exists a waiting small job }{
                        $j \gets$ the smallest small job \;  
                        Schedule $j$ on $i$\;
                    }
                }
                \For {each idle \All machine $i$} {
                    \If{ there exists a waiting job that is not proxied }{
                        $j \gets$ the smallest waiting unproxied job \;
                        Schedule $j$ on machine $i$ \;
                        $\phi \gets 0$ if $j$ is an active large job \;
                    }
                }
        }
\end{algorithm}

The basic invariant we aim to keep in the algorithm is \Cref{lem:restart_block_num}.
\begin{lemma}
\label{lem:restart_block_num}
The number of waiting blocked small jobs at any time is at most $\sqrt{nm}$.
\end{lemma}

\begin{proof}
We partition the timeline into intervals, where each interval begins when the algorithm starts processing an active large job and ends just before the next such event. By the algorithm, the counter $\phi$ is reset to $0$ at the start of each interval. Within an interval, every newly released blocked small job increments $\phi$ by one, and when $\phi$ reaches $\lfloor\sqrt{nm}\rfloor$, all currently processing active large jobs are killed and $\phi$ is reset. Therefore, at most $\lfloor\sqrt{nm}\rfloor$ blocked small jobs are created in each interval.

It remains to show that all blocked small jobs created in one interval have started processing before the next interval begins. After a kill event (or if the active large jobs complete naturally), no active large job is being processed. At this point, every waiting blocked small job is smaller than every waiting active large job, so by the shortest-first scheduling rule, all blocked small jobs are scheduled before any active large job can start processing again.
\end{proof}

Moving into the analysis, we now distinguish between several types of job sets: The \emph{real} job set $J$, which includes only actual jobs and excludes proxy jobs. The \emph{unblocked small job set} $S$, and the \emph{blocked small job set} $K$, both of which may include proxy jobs. The \emph{large job set} $L$, which includes all large jobs, even if they become retired at some point.

Note that $K \cup S \cup L \not\subseteq J$ with the presence of proxy jobs. We analyze the total flow time contributed by $S$, $K$, and $L$ separately.

\begin{lemma}[Unblocked Small Job]
\label{lem:dunblocked-small-job}
Let $S$ be the set of unblocked small jobs (including proxy jobs). Then:
\[
\sum_{j \in S} F_j = O\left( \sqrt{\frac{n}{m}} \right) \cdot \OPT.
\]
\end{lemma}

\begin{proof}
In the absence of blocked small jobs and committed large jobs, the total flow time contributed by the jobs in $S$ is simply $F(\NSJF(S))$, which is at most $O\!\left(\sqrt{n/m}\right)\cdot F(\OPT(S))$ by \Cref{lem:sjfmain}. Indeed, active large jobs never interfere with the scheduling of unblocked small jobs: if a job $j$ is unblocked, then no active large job is running at its release time, and no active large job can be started before $j$, since every such job is larger than $j$.

We now account for the effect of blocked small jobs and committed large jobs. 

To analyze their effect, we partition the time line into intervals
\[
I_h=[t_h^{(\ell)},t_h^{(r)}), \qquad h=1,2,\dots,H,
\]
defined recursively as follows. We let $t_1^{(\ell)}$ be the first time at which an unblocked small job is released. Given $t_h^{(\ell)}$, we let $t_h^{(r)}$ be the first time at or after $t_h^{(\ell)}$ at which the algorithm starts processing an active large job, which may later be killed or retired. If no such time exists, we set $t_h^{(r)}=\infty$ and the construction completes. Otherwise, we let $t_{h+1}^{(\ell)}$ be the first time after $t_h^{(r)}$ at which an unblocked small job is released, and continue recursively. By construction, every unblocked small job is released in some interval $I_h$.

The key properties of each interval $I_h$ are the following. First, no active large job is processed inside $I_h$: 
by definition, no large job is processing at $t_h^{(\ell)}$ since an unblocked small job is released at $t_h^{(\ell)}$, and no large job starts processing inside $I_h$. Second, every blocked small job released before $t_h^{(\ell)}$ must start processing before $t_h^{(r)}$: the algorithm prioritizes small jobs over the large job starting at $t_h^{(r)}$.

For each $h$, let $S_h$ be the set of unblocked small jobs released in $I_h$, and let $K_h$ be the set of blocked small jobs that started processing in $I_h$. For every job in $K_h$, we modify its release time to its start time and obtain the corresponding set $K_h'$.

Shifting the time origin to $t_h^{(\ell)}$, the schedule of $S_h \cup K_h'$ inside $I_h$ can be viewed as $\NSJF(S_h \cup K_h', \vec{b_h})$, where $\vec{b_h}$ is the initial blocking vector: machine $i$ is unavailable during $[0, b_{h,i})$ due to a small job or a committed large job still being processed at time $t_h^{(\ell)}$, with $b_{h,i} = 0$ if machine $i$ is idle. Therefore,
\[
\sum_{j\in S_h} F_j 
\le
\sum_{j\in S_h \cup K'_h} F_j 
=
F(\NSJF(S_h\cup K_h',\vec{b_h})).
\]
We now apply the generalized \NSJF\ bound from \Cref{lem:SJFgeneral}. Since every small job or committed large job has size at most $O(\OPT/\sqrt{nm})$, the total blocking time $B_h = \sum_i b_{h,i}$ satisfies $B_h \leq m \cdot O(\OPT/\sqrt{nm}) = O(\sqrt{m/n} \cdot \OPT)$. Hence,
\[
F(\NSJF(S_h\cup K_h',\vec{b_h}))
\le
F(\OPT^m(S_h\cup K_h'))
+
O\!\left(\frac{|S_h|+|K_h'|}{\sqrt{nm}}\right)\OPT,
\]
where $\OPT^m$ denotes the optimal preemptive schedule \emph{with} migration.

It remains to compare $F(\OPT^m(S_h\cup K_h'))$ with $F(\OPT(S_h))$. For this, consider an optimal preemptive schedule $\OPT(S_h)$. We insert the jobs in $K_h'$ into $\OPT(S_h)$ one by one, in a migratory fashion: for each job $j\in K_h'$, we scan forward from $r_j$ and process $j$ at the first available idle slot without delaying the original schedule. If $j$ overlaps a busy period, we interrupt it and resume it at the next idle slot, possibly on a different machine.

Let $F_j'$ denote the resulting flow time of job $j \in K_h'$ in this schedule. Since $j$ is only delayed during busy periods, its waiting time $F_j' - p_j$ is at most the total busy time, which is $\sum_{i \in S_h \cup K_h'} p_i / m$.

It remains to bound $|K_h'|$. By the restart rule, if a job belongs to $K_h$, then it must have been released while some machine was processing a currently unproxied large job, and hence before $t_h^{(\ell)}$. Therefore, every job in $K_h$ is already waiting at time $t_h^{(\ell)}$. By \Cref{lem:restart_block_num}, the number of waiting blocked small jobs is at most $\sqrt{nm}$, so 
\[
|K_h| = |K_h'| \le \sqrt{nm}.
\]
It follows that
\[
F(\OPT^m(S_h \cup K_h'))
\le
F(\OPT(S_h))
+
O\!\left(\sqrt{\frac{n}{m}}\right)\sum_{j \in S_h \cup K_h'} p_j.
\]

Combining the above bounds and summing over all $h$, we get
\begin{align*}
\sum_{j\in S} F_j
&\le \sum_{h=1}^H F(\NSJF(S_h\cup K_h',\vec{b_h})) \\
&\le \sum_{h=1}^H F(\OPT^m(S_h\cup K_h'))
   + O\!\left(\sum_{h=1}^H \frac{|S_h|+|K_h'|}{\sqrt{nm}}\right)\OPT \\
&\le \sum_{h=1}^H F(\OPT(S_h))
   + O\!\left(\sum_{h=1}^H \frac{|S_h|+|K_h'|}{\sqrt{nm}}\right)\OPT \\
&\quad + O\!\left(\sqrt{\frac{n}{m}}\right)
   \sum_{h=1}^H \sum_{j\in S_h\cup K_h'} p_j.
\end{align*}

Finally, the sets $\{S_h\}_{h=1}^H$ partition $S$, and the sets $\{K_h'\}_{h=1}^H$ partition $K'$. Thus,
\[
\sum_{h=1}^H F(\OPT(S_h)) \le F(\OPT(S)),
\qquad
\sum_{h=1}^H (|S_h|+|K_h'|) \le n.
\]
The total processing volume of the jobs in $\bigcup_h (S_h\cup K_h')$ is at most that of the original jobs in $J$, giving
\[
\sum_{h=1}^H \sum_{j\in S_h\cup K_h'} p_j \le O(\OPT).
\]
Therefore,
\[
\sum_{j \in S} F_j
\le
O\!\left(\sqrt{\frac{n}{m}}\right)\OPT.
\]
\end{proof}

\begin{lemma}[Blocked Small Job]
\label{lem:dblocked-small-job}
Let $K$ be the set of blocked small jobs. Then the total flow time of jobs in $K$ is bounded by:
\[
\sum_{j \in K} F_j \leq O\left( \sqrt{\frac{n}{m}} \right) \cdot \OPT.
\]
\end{lemma}

\begin{proof}
Consider an arbitrary blocked small job $j$. By our scheduling rule, $j$ can only experience waiting time when all \SmallOnly machines are busy, since $j$ is classified as a small job. During the entire interval $[r_j, s_j)$, all \SmallOnly machines must be fully occupied. We refer to this as a \emph{busy time} of \SmallOnly machines.

Due to the restart rule, at any busy time $t$ of \SmallOnly machines, the number of blocked small jobs still waiting is at most $\sqrt{nm}$ by \Cref{lem:restart_block_num}.

Since there are $\lfloor m/2 \rfloor$ \SmallOnly machines and they are fully occupied during busy periods, the total length of all busy intervals is at most
\[
\frac{2 \cdot \OPT}{m - 1}.
\]
Therefore, if $m \geq 2$, the total waiting time of all blocked small jobs is bounded by
\[
O\left( \sqrt{\frac{n}{m}} \right) \cdot \OPT.
\]
We remark that when $m = 1$, there is no \SmallOnly machine, and this argument no longer holds. Finally, the total flow time of all blocked small jobs is:
\[
\sum_{j \in K} F_j = \sum_{j \in K} (s_j - r_j + p_j) \leq O\left( \sqrt{\frac{n}{m}} \right) \cdot \OPT.
\]

\end{proof}

For each large job $j$,
we only need to analyze the modified flow time:
\[
\hat{F}_j = 
\begin{cases}
r_{j'} - r_j & \text{if $j$ is proxied by $j'$} \\
C_j - r_j & \text{if $j$ is unproxied}
\end{cases}.
\]

\begin{lemma}[Large Job]
\label{lem:dlarge-job}
    For the large-job set $L$, we have $\sum_{j\in L} \hat{F}_j \leq O(\sqrt{n / m}) \cdot \OPT$.
    \end{lemma}
\begin{proof}
We focus on the delay interval $[r_j, t_j)$, where
$$
t_j = \begin{cases}
r_{j'} & \text{if $j$ is proxied by $j'$,} \\
s_j & \text{otherwise,}
\end{cases}
$$
and aim to bound $t_j - r_j$. We begin by showing that shifting the release time $r_j$ by $\sqrt{n/m} \cdot p_j$ is acceptable. Let $t'_j = r_j + \sqrt{n/m} \cdot p_j$. Then:
\begin{equation}
\label{eqn:restart_large_shifting}
\sum_{j \in L} (t'_j - r_j) = \sqrt{\frac{n}{m}} \cdot \sum_{j \in L} p_j \leq \sqrt{\frac{n}{m}} \cdot \OPT.
\end{equation}

We now analyze the remaining delay interval $[t'_j, t_j)$. For each time $t \in [t'_j, t_j)$, we classify it into one of the following two categories:
\begin{enumerate}
    \item[$(\checkmark)$] \textbf{Genuine busy time:} All \All machines are busy at $t$, and at least $\lceil m/4 \rceil$ of the jobs running on \All machines are not restarted later.
    \item[$(\times)$] \textbf{Fake busy time:} All \All machines are busy at $t$, but more than $\lceil m/4 \rceil$ of the jobs running on \All machines will be restarted later.
\end{enumerate}

We decompose the delay $t_j - t'_j$ into two parts: $F_j^{(\checkmark)}$ and $F_j^{(\times)}$, corresponding to time spent in genuine busy and fake busy intervals, respectively:
\[
t_j - t'_j = F_j^{(\checkmark)} + F_j^{(\times)}.
\]

\paragraph{Bounding $F_j^{(\checkmark)}$ (genuine busy).}  
At any genuine busy time, at least $\lceil m/4 \rceil$ machines are processing jobs that are not restarted. The total duration of such intervals is at most $O(\OPT / m)$. Since there are at most $\ell = \lfloor\sqrt{nm}\rfloor$ large jobs at any time by \Cref{lem:partition_restart}, we have:
\begin{equation}
\label{eqn:restart_large_busy}
\sum_{j \in L} F_j^{(\checkmark)} \leq \sqrt{nm} \cdot O\left( \frac{\OPT}{m} \right) = O\left( \sqrt{\frac{n}{m}} \cdot \OPT \right).
\end{equation}

\paragraph{Bounding $F_j^{(\times)}$ (fake busy).}  
Next, we analyze the contribution of fake busy intervals induced by the kill events. Let the $i$-th kill event occur at time $\beta_i$. The key observation is the following.

\begin{claim}
\label{claim:fake_busy}
For any $t \in [t'_j, t_j)$ that is a fake busy time between consecutive kill events $\beta_{i-1}$ and $\beta_i$, we have $\beta_i - t \leq p_j$.
\end{claim}
\begin{proof}[Proof of \Cref{claim:fake_busy}]
Since $t$ is a fake busy time, more than $\lceil m/4 \rceil$ jobs are restarted at $\beta_i$. Suppose, for contradiction, that $\beta_i - t > p_j$. Then these jobs must have been running continuously since $t$, implying that their processing times are at least $p_j$. By our algorithm, job $j$ had not yet been released when the algorithm started one of these larger jobs. Therefore, letting $U$ denote this set of jobs, we have:
\[
\beta_i < r_j + \min_{u \in U} p_u.
\]
However, by the partition rule $p_j \geq 4P(r_j)/\ell$ and $\ell = \lfloor\sqrt{nm}\rfloor$, we have $t'_j = r_j + \sqrt{n/m} \cdot p_j \geq r_j + 4P(r_j)/m$. Since the jobs in $U$ were released before $r_j$, their total size is at most $P(r_j)$. Moreover, $|U| > \lceil m/4 \rceil$, so:
\[
\frac{4\sum_{u \in U} p_u}{m} \geq \frac{4 \cdot |U| \cdot \min_{u \in U} p_u}{m} > \min_{u \in U} p_u.
\]
Combining, $$t'_j \geq r_j + 4P(r_j)/m \geq r_j + 4\sum_{u \in U} p_u / m > r_j + \min_{u \in U} p_u > \beta_i.$$
This gives $\beta_i - t \leq \beta_i - t'_j < 0$, a contradiction.
\end{proof}

By \Cref{claim:fake_busy}, each large job $j$ contributes at most $p_j$ to $F_j^{(\times)}$ per kill event. Thus, for each kill event, the total contribution to $\sum_j F_j^{(\times)}$ is at most:
\[
\sum_{j \in L} p_j \leq \OPT.
\]

Finally, we bound the number of kill events. Since each kill is triggered only after at least $\lfloor \sqrt{nm} \rfloor$ new blocked small jobs are released, the total number of kill events is at most $O(\sqrt{n/m})$. Therefore:
\begin{equation}
\label{eqn:restart_large_perrestart}
\sum_{j \in L} F_j^{(\times)} \leq \sqrt{\frac{n}{m}} \cdot \OPT.
\end{equation}

Combining \eqref{eqn:restart_large_shifting}, \eqref{eqn:restart_large_busy}, and \eqref{eqn:restart_large_perrestart}, we obtain:
\[
\sum_{j \in L} \hat{F}_j \leq \sum_{j \in L} \left( p_j + t_j - t'_j + t'_j - r_j \right) = O\left( \sqrt{\frac{n}{m}} \right) \cdot \OPT.
\]
\end{proof}

\begin{theorem}
\label{thm:det_restart}
    \Cref{alg:multi-deterministic-restart-m} is an online polynomial-time deterministic algorithm with kill-and-restart for total flow time minimization that is $O(\sqrt{n/m})$-competitive against the preemptive offline solution, provided that $m\geq 2$.
\end{theorem}
\begin{proof}
    By \Cref{lem:dunblocked-small-job}, \Cref{lem:dblocked-small-job}, and \Cref{lem:dlarge-job}.
    \[
    \sum_{j \in J} F_j \leq \sum_{j \in S} F_j + \sum_{j \in K} F_j + \sum_{j \in L} \hat{F}_j = O\left( \sqrt{\frac{n}{m}} \right) \cdot \OPT. 
    \]
\end{proof}

\newpage

\section{Deterministic Algorithm with Kill-and-Restart without the Knowledge of \texorpdfstring{$n$}{n}}
\label{sec:unknownn_det_restart}

In this section, we extend the $O(\sqrt{n/m})$-competitive kill-and-restart algorithm of \Cref{sec:det-restart} to the setting where the total number of jobs $n$ is unknown in advance. The resulting algorithm achieves a competitive ratio of $O(n^{\alpha}/\sqrt{m})$, where $\alpha = (\sqrt{5}-1)/2$.

The overall framework is the same: we keep the machine partition into \SmallOnly\ and \All\ machines, and we still kill large jobs once enough small jobs are blocked. The key challenge is in the partition. Recall that in the known-$n$ setting, the threshold $\ell$ was fixed, making the partition nearly stable: a small job could never become large (since $\ell$ does not grow), and the only reclassification was in the opposite direction --- a large job could be retired and handled via a proxy job.

Without knowledge of $n$, the threshold $\ell(t)$ must depend on the number of jobs $n(t)$ released so far, and it increases over time. This means a job's classification can oscillate: for example, a job of rank $4$ is small when $\ell = 3$, becomes large when $\ell$ grows to $4$, and may become small again when $\ell$ reaches $5$ and larger jobs take its place. This instability invalidates the previous analysis.

To address this, we make two modifications, described below and formalized in \Cref{alg:multi-deterministic-restart-m-unknownn}.

\paragraph{Dynamic partition rule.}
We round $n(t)$ up to the nearest power of two, $N = 2^{\lceil \log_2 n(t) \rceil}$, and set the threshold $\ell = \lfloor N^{\alpha}\sqrt{m} \rfloor$. At every time $t$, a job $j$ is classified as \emph{large} if it is among the top $\ell$ jobs by processing time and $p_j \ge 4P(t)/\ell$, where $P(t)$ is the total processing time of all jobs released by time $t$; otherwise, it is classified as \emph{small}. The rounding ensures that $\ell$ changes only $O(\log n)$ times.

\paragraph{Additional kill event on threshold update.}
When $N$ doubles, some jobs currently running on \SmallOnly machines may now be classified as large (since $\ell$ has increased and new jobs may have displaced them in rank). We kill all such jobs immediately. This is a new type of kill event, in addition to the existing one triggered when too many small jobs are blocked.

\begin{algorithm} [htbp]
    \caption{Deterministic Algorithm with Kill-and-Restart on $m$ Machines, without the knowledge of $n$}
    \label{alg:multi-deterministic-restart-m-unknownn}
    \KwData{Estimate of $n$: $N$ initialized to $1$; }
    \SetKwBlock{OnJobRelease}{\textbf{On job $j$ release at $t$:}}{}
    \SetKwBlock{OnEvent}{\textbf{On any job completion or after any job release:}}{}
    \OnJobRelease
    {
        \If{the number of jobs exceeds $N$}
        {
            $N \gets 2N$\;
            Mark $N$ as \textbf{updated} \;
        }
        \tcp{Update the Partition of all jobs again based on the current number of jobs $N$ and the current sum of processing time $P$.}
        $\ell =  \lfloor N^{\alpha}\sqrt{m} \rfloor$ where $\alpha = (\sqrt{5}-1)/2$ \;
        Call every job $i
        \begin{cases}
            \textbf{currently large} & \text{it is a top-$\ell$ large job and $p_i \geq 4P/\ell$} \\
            \textbf{currently small} & \text{otherwise} 
        \end{cases} $\;
        \If{$N$ is \textbf{updated} by the release of job $j$}
        {
           Kill all currently large jobs processing on \SmallOnly machines \;
        }
        \If{there exists a processing job which is currently large} 
        {
            $\phi \gets$ the number of waiting jobs that are currently small \; 
            \If{$\phi > \ell$}
            {
                Kill all processing jobs that are currently large. 
            }
        }
    }
    \OnEvent
        {  
            \For {each idle \SmallOnly machine $i$} {
                \If{ there exists a waiting currently small job }{
                    $j \gets$ the smallest currently small job \;  
                    Schedule $j$ on $i$\;
                }
            }
            \For {each idle \All machine $i$} {
                \If{ there exists a waiting job }{
                    $j \gets$ the smallest waiting job \;
                    Schedule $j$ on $i$ \;
                }
            }
        }
\end{algorithm}

\medskip
Next, we move on to the analysis. 
For each job $j$, let $T_j(t) \in \{\text{Large}, \text{Small}\}$ denote the state of job $j$ at time $t$.
Let $C_n$ be the completion time of the last completed job, which marks the end of the instance.  
Define $\theta_j$ as the earliest time such that $T_j(t) = \text{Small}$ for all $t \in [\theta_j, C_n)$.
The interval $[\theta_j, C_n)$ is called the \emph{last small period} of job $j$.

We divide the flow time of each job $j$ into waiting time and its processing time, i.e., $[r_j,s_j) \cup [s_j,C_j)$. Note that to prove the total flow time is in $O(n^\alpha/\sqrt{m}) \cdot \OPT$, we only need to discuss the total waiting time, since the total processing time is bounded by $\OPT$. Furthermore, we divide the waiting time of each job into two parts with respect to the last small periods as follows. 

\begin{itemize}
    \item \textbf{Large waiting time} ($F_j^{(1)}$): the waiting time before the last small period, i.e.,
    \[
    F_j^{(1)} = |[r_j, s_j) \cap [r_j, \theta_j)|.
    \]
    \item \textbf{Small waiting time} ($F_j^{(2)}$): the waiting time during the last small period, i.e.,
    \[
    F_j^{(2)} = |[r_j, s_j) \cap [\theta_j, C_n)|.
    \]
\end{itemize}
By this definition, the total flow time of job $j$ is
\[
F_j = F_j^{(1)} + F_j^{(2)} + p_j.
\]
If job $j$ is a \emph{final large job} (i.e., $T_j(t) = \text{Large}$ at the end of the instance),
then it has no last small period, and thus its entire waiting time equals its large waiting time:
\[
F_j = F_j^{(1)}, \quad F_j^{(2)} = 0.
\]

Note that in the version where $n$ is known, a large job can change its state to \emph{small} at most once, and a small job can never become \emph{large}. 
In that setting, we introduced the notion of a \emph{proxy job} to decompose the flow time. 
A similar idea applies here: if a large job becomes small exactly once, then the proxy job's waiting time corresponds precisely to the \emph{small waiting time} $F_j^{(2)}$ defined above, while the remaining part of the original job's waiting time corresponds to the \emph{large waiting time} $F_j^{(1)}$.

Then we bound $F_j^{(1)}$ and $F_j^{(2)}$ separately. 
The analysis for $F_j^{(1)}$ is similar to the case where $n$ is known. 
Note that we have two types of kill events. 
One occurs when \emph{$N$ is updated}, which we call a \textbf{type-A} event; 
the other occurs when \emph{$\phi$ becomes large}, which we call a \textbf{type-B} event. 
The number of type-B kill events can be bounded similarly to the known-$n$ case, 
with the help of the rounded benchmark value $N$.

\begin{lemma}
    The number of type-B kill events in the algorithm is at most $O(n^{1-\alpha}/\sqrt{m})$.
\end{lemma}
\begin{proof}
We know that if a type-B kill event occurs at time $t$, there must be at least $\lfloor N(t)^{\alpha} \sqrt{m} \rfloor$ small jobs waiting, where $N(t)$ denotes the value of $N$ at time $t$. We interpret these small jobs as the creators of this kill event.

Note that when a small job $j$ creates a type-B kill event at time $t$, it remains small until $N$ is updated. Hence, after the kill event at $t$ and before the next update of $N$, no large job will be processed on any machine unless we start job $j$. On the other hand, if we start $j$, it will not be killed again until $N$ is updated, since $j$ stays small during this period. Therefore, job $j$ cannot trigger another type-B kill event before $N$ changes.

Consequently, the number of kill events that occur while $N = x$ is at most $x^{1-\alpha}/\sqrt{m}$. Since $N$ increases geometrically by a factor of $2$, the total number of kill events is bounded by
$
O\!\left(\frac{n^{1-\alpha}}{\sqrt{m}}\right),
$
as $N \le 2n$.
\end{proof}

Using the upper bounds of the number of kill events, we are able to provide an upper bound for the large waiting time of all jobs (i.e., $F_j^{(1)}$). 

\begin{lemma}
\label{lem:unknownn_large_waiting}
For the large waiting time, we have
\[
\sum_{j \in J} F^{(1)}_j \le O\!\left(\frac{n^{\alpha}}{\sqrt{m}}\right) \OPT.
\]
\end{lemma}
\begin{proof}
We first show that, for any time $t$, the number of released jobs that are not in their last small period is at most $O(n^{\alpha}\sqrt{m})$.  
Indeed, if a job $j$ is not in its last small period at time $t$, it must be large at some later time $t' > t$. 
This means that at time $t'$, $j$ ranks among the top-$\lfloor (2n(t'))^{\alpha}\sqrt{m} \rfloor$ largest jobs within a larger set of released jobs. 
Tracing back to the earlier time $t$, since the released job set is smaller, $j$ must also be among the top-$\lfloor (2n(t'))^{\alpha}\sqrt{m} \rfloor$ largest jobs at time $t$, 
and therefore also among the top-$\lfloor (2n)^{\alpha}\sqrt{m} \rfloor$ largest jobs because $n > n(t')$. 
Hence, at any time $t$, the number of released jobs not in their last small period is bounded by $O(n^{\alpha}\sqrt{m})$.

Next, we follow the same argument as in \Cref{lem:dlarge-job}.  
We shift each job’s release time to
\[
t'_j = r_j + \frac{(2n)^{\alpha}}{\sqrt{m}} \cdot p_j,
\]
and bound the waiting time starting from $t'_j$. It is enough because we have
\begin{equation}
\label{eqn:unknownn_restart_large_shifting}
\sum_{j \in J} (t'_j - r_j)
= \frac{(2n)^{\alpha}}{\sqrt{m}} \sum_{j \in J} p_j
\le \frac{(2n)^{\alpha}}{\sqrt{m}} \cdot \OPT.
\end{equation}
If the total waiting time after the shifted times $t'_j$ is already bounded by $O(n^{\alpha}/\sqrt{m}) \OPT$, then $\sum_{j \in J} F^{(1)}_j$ is also bounded by $O(n^{\alpha}/\sqrt{m}) \OPT$.

For each job $j$, let its large waiting interval be $[t'_j, t_j)$.  
For every time $t \in [t'_j, t_j)$, we classify it into two types:

\begin{enumerate}
    \item[$(\checkmark)$] \textbf{Genuine busy time:} All \All\ machines are busy at $t$, and at least $\lceil m/4 \rceil$ of the jobs running on \All\ machines are not restarted later.
    \item[$(\times)$] \textbf{Fake busy time:} All \All\ machines are busy at $t$, but more than $\lceil m/4 \rceil$ of the jobs running on \All\ machines will be restarted later.
\end{enumerate}

We decompose the delay $t_j - t'_j$ into two parts: $F_j^{(\checkmark)}$ and $F_j^{(\times)}$, corresponding to time spent in genuine busy and fake busy intervals, respectively:
\[
t_j - t'_j = F_j^{(\checkmark)} + F_j^{(\times)}.
\]

\paragraph{Bounding $F_j^{(\checkmark)}$ (genuine busy).}
During genuine busy intervals, at least $\lceil m/4 \rceil$ machines are processing jobs that are never restarted.  
The total duration of such intervals is at most $O(\OPT / m)$.  
Each interval can contribute to the waiting time of at most $O(n^{\alpha}\sqrt{m})$ different large jobs.  
Hence,
\begin{equation}
\label{eqn:unknownn_restart_large_busy}
\sum_{j \in J} F_j^{(\checkmark)} \le O\!\left(\frac{n^{\alpha}}{\sqrt{m}} \cdot \OPT \right).
\end{equation}

\paragraph{Bounding $F_j^{(\times)}$ (fake busy).}
Now consider fake busy intervals. These must be caused by type-B kill events, since type-A kill events only kill jobs on \SmallOnly machines. Let the $i$-th type-B kill event occur at time $\beta_i$. The key observation is the following.

\begin{claim}
\label{claim:unknownn_fake_busy}
For any $t \in [t’_j, t_j)$ that is a fake busy time between consecutive type-B kill events $\beta_{i-1}$ and $\beta_i$, we have $\beta_i - t \le p_j$.
\end{claim}
\begin{proof}[Proof of \Cref{claim:unknownn_fake_busy}]
Since $t$ is a fake busy time, more than $\lceil m/4 \rceil$ jobs are restarted at $\beta_i$.
Suppose, for contradiction, that $\beta_i - t > p_j$.
Then these jobs must have been running continuously since $t$, implying that their processing times are at least $p_j$.
By our algorithm, job $j$ had not yet been released when the algorithm started one of these larger jobs.
Therefore, letting $U$ denote this set of jobs, we have
\[
\beta_i < r_j + \min_{u \in U} p_u.
\]
However, by the partition rule $p_j \geq 4P(r_j)/\ell$ and $\ell = \lfloor N^{\alpha}\sqrt{m} \rfloor$, we have $t’_j = r_j + (2n)^{\alpha}/\sqrt{m} \cdot p_j \geq r_j + 4P(r_j)/m$. Since the jobs in $U$ were released before $r_j$, their total size is at most $P(r_j)$. Moreover, $|U| > \lceil m/4 \rceil$, so:
\[
\frac{4\sum_{u \in U} p_u}{m} \geq \frac{4 \cdot |U| \cdot \min_{u \in U} p_u}{m} > \min_{u \in U} p_u.
\]
Combining, $t’_j \geq r_j + 4P(r_j)/m \geq r_j + 4\sum_{u \in U} p_u / m > r_j + \min_{u \in U} p_u > \beta_i$.
This gives $\beta_i - t \leq \beta_i - t’_j < 0$, a contradiction.
\end{proof}

Each kill event therefore contributes at most $\sum_{j \in J} p_j \le \OPT$ to $\sum_j F_j^{(\times)}$.  
Since the total number of type-B kill events is at most $O(n^{1-\alpha}/\sqrt{m})$, we obtain
\begin{equation}
\label{eqn:unknownn_restart_large_perrestart}
\sum_{j \in J} F_j^{(\times)} \le O\!\left(\frac{n^{1-\alpha}}{\sqrt{m}}\right) \cdot \OPT \leq O\!\left(\frac{n^{\alpha}}{\sqrt{m}}\right)\cdot \OPT.
\end{equation}
The last inequality holds because $\alpha$ is set to $(\sqrt{5}-1)/2 > 0.5$.

Finally, combining
\eqref{eqn:unknownn_restart_large_shifting},
\eqref{eqn:unknownn_restart_large_busy}, and
\eqref{eqn:unknownn_restart_large_perrestart}, we have
\[
\sum_{j \in J} F^{(1)}_j
\le \sum_{j \in J} \bigl( (t'_j - r_j) + F_j^{(\checkmark)} + F_j^{(\times)} \bigr)
\le O\!\left(\frac{n^{\alpha}}{\sqrt{m}}\right) \cdot \OPT.
\]
\end{proof}

Next, we aim to bound the total small waiting time. We define an increasing sequence of time points $t_h$, each corresponding to the start time of a job $j$ that is not in its last small period (i.e. $t_h < \theta_j$, and it is possible to be killed later). We further define $t_0 = 0$ and $t_{\bar{h}+1} = C_n$, where $t_{\bar{h}}$ is the last time we start processing a job that is not in its last small period and $C_n$ is the last completion time. These $t_h$ allow us to divide the jobs into sets, corresponding to $\theta_j$.

$$
S_h = \{j \in J \mid \theta_j \in [t_{h}, t_{h+1}) \text{ and } F^{(2)}_j > 0\}.
$$
Some basic properties hold for every $j \in S_h$:
\begin{enumerate}
    \item $j$ is started within the time interval $(\theta_j, t_{h+1})$.
    \item $j$ will not be killed after this start. 
\end{enumerate}
We now justify these two properties in order.
For the first one, if $h = \bar{h}$, the property holds trivially; otherwise, since $t_{h+1}$ is the start time of another job $j'$, where $\theta_{j'} > t_{h+1}$. By definition, $T_{j'}(t)$ is still large for some 
$t > t_{h+1}$, while $T_j(t)$ is small, implying $p_{j'} > p_{j}$. Given that our algorithm prioritizes scheduling smaller waiting jobs, job $j$ must be started before $j'$ starts at $t_{h+1}$. On the other hand, since $j$ is small after $\theta_j$, if it is started in $[t_h,t_{h+1})$ after $\theta_j$, it will not be killed. 

To bound jobs' small waiting time in every $S_h$, we further divide the jobs into two types.

\begin{enumerate}
    \item \textbf{Blocked:} $j \in S_h$, and some machine is processing a job $j'$ at $\theta_j$, where $\theta_j < \theta_{j'}$. We use $K_h$ to denote the set of blocked small jobs in $S_h$. 
    \item \textbf{Unblocked:} otherwise, and we use $U_h$ to denote the set of unblocked jobs in $S_h$.
\end{enumerate}

Let $t^*$ be the last time in $[t_h, t_{h+1})$ during which some machine processes a job $j$ with $\theta_j > t^*$.  
If no such time exists, we set $t^* = t_{h+1}$, which implies $U_h = \emptyset$.   

\begin{lemma}
\label{lem:unknownn-crowded}
If there exists a time in $[t_h, t_{h+1})$ at which the number of waiting blocked jobs exceeds $n^{\alpha}\sqrt{m}$, then
\[
\sum_{j \in S_h} F^{(2)}_j \le F(\OPT(S_h)) + O\!\left(\frac{|S_h|}{n^{1-\alpha}\sqrt{m}}\right) \cdot \OPT.
\]
\end{lemma}
\begin{proof}
Since after $t^*$ the number of waiting blocked jobs can only decrease, there must exist a time $t \in [t_h, t^*)$ such that the number of waiting blocked jobs exceeds $n^{\alpha}\sqrt{m}$.  
All blocked jobs belong to $S_h$, hence they are in their last small period and are guaranteed to be small at time $t$.  
The only reason our algorithm does not kill the processing job $j$ that completes at $t^*$ with $\theta_j > t^*$ is that $T_j(t)$ is still small.  
Therefore, its size is bounded by $\frac{\OPT}{N(t)^{\alpha}\sqrt{m}}$.  

Although $N$ may be smaller than $n$ since $n$ is unknown at that time, it is lower bounded by $n^{\alpha}\sqrt{m}/2$, because at least $n^{\alpha}\sqrt{m}$ small blocked jobs are waiting, implying $n(t) \ge n^{\alpha}\sqrt{m}$.  
Hence, the size of $j$ satisfies
\[
p_j \le \frac{\OPT}{N(t)^{\alpha}\sqrt{m}} \le \frac{2\OPT}{n^{\alpha^2}m^{\alpha/2}\sqrt{m}} \le \frac{2\OPT}{n^{\alpha^2}\sqrt{m}}.
\]

We now view the schedule of $S_h$ in our algorithm as running \NSJF on a shifted instance $S'_h$, with an initial blocking time on each machine, where $S'_h$ is obtained from $S_h$ by shifting all release times by $-t_h$ (it is the same as viewing $t_h$ as time $0$).

Each machine may be initially blocked because it is processing some job at time $t_h$, either by a job not in its last small period (for at most $t^*-t \le \frac{2\OPT}{n^{\alpha^2}\sqrt{m}}$) or by a job already in its last small period (for at most $\frac{\OPT}{n^{\alpha}\sqrt{m}} \le \frac{2\OPT}{n^{\alpha^2}\sqrt{m}}$ since $\alpha<1$). 
Hence every component of the blocking vector $\vec{b}$ is at most $\frac{2\OPT}{n^{\alpha^2}\sqrt{m}}$. 
It follows that the schedule of $S_h$ in our algorithm coincides with $\NSJF(S'_h,\vec{b})$.

By \Cref{lem:SJFgeneral}, writing $\vec{b}$ for the initial blocking vector,
\begin{align*}
    \sum_{j \in S_h} F^{(2)}_j 
    \;\le\; F(\NSJF(S'_h,\vec{b}))
    &\le\; F(\OPT(S_h))
      \;+\; 2|S_h|\cdot \frac{\OPT}{n^{\alpha}\sqrt{m}}
      \;+\; |S_h|\cdot \frac{2\OPT}{n^{\alpha^2}\sqrt{m}} \\
    &\le\; F(\OPT(S_h)) + O\!\left(\frac{|S_h|}{n^{\alpha^2}\sqrt{m}}\right)\!\cdot \OPT \\
    &=\; F(\OPT(S_h)) + O\!\left(\frac{|S_h|}{n^{1-\alpha}\sqrt{m}}\right)\!\cdot \OPT,
\end{align*}
where the last equality uses $\alpha^2 = 1-\alpha$ for $\alpha=(\sqrt{5}-1)/2$.
\end{proof}

\begin{lemma}
\label{lem:unknownn-uncrowded}
If for all $t \in [t_h, t_{h+1})$ the number of waiting blocked jobs is at most $n^{\alpha}\sqrt{m}$, then
\[
\sum_{j \in U_h} F^{(2)}_j \le F(\OPT(S_h))
+ O\!\left(\frac{n^{\alpha}}{\sqrt{m}}\right) \sum_{j \in S_h} p_j
+ O\!\left(\frac{|S_h|}{n^{\alpha}\sqrt{m}}\right) \cdot \OPT.
\]
\end{lemma}
\begin{proof}
We focus on the number of waiting blocked jobs at time $t^*$, which is at most $n^{\alpha}\sqrt{m}$. Let $K^*_h$ denote this set. 
View $t^*$ as time $0$ and form two shifted instances $U'_h$ and $K'_h$ by moving every release time in $U_h$ and $K^*_h$ by $-t^*$, respectively.
The schedule of $U_h$ in our algorithm is a shifting to the schedule of $U'_h$ under $\NSJF(U'_h \cup K'_h,\vec{b})$, 
where $\vec{b}$ is the initial blocking vector from the jobs processing at $t^*$. 
By the definition of $t^*$, these processing jobs are already in their last small period, so each component of $\vec{b}$ is at most $O\!\left(\frac{\OPT}{n^{\alpha}\sqrt{m}}\right)$.

By \Cref{lem:SJFgeneral},
\begin{equation}
\label{eqn:unknownn-unblocked2-eqn1}
\sum_{j \in U_h} F^{(2)}_j 
\;\le\; F(\NSJF(U'_h \cup K'_h,\vec{b}))
\;\le\; F(\OPT^m(U'_h \cup K'_h))
\;+\; O\!\left(\frac{|U'_h \cup K'_h|}{n^{\alpha}\sqrt{m}}\right)\OPT,
\end{equation}
where $\OPT^m$ denotes the optimal preemptive, migratory schedule.


The next step is to bound the gap between $\OPT^m(U'_h \cup K'_h)$ and $\OPT^m(U_h)$ using the fact that $|K'_h|$ is small. 
Consider an optimal preemptive, migratory schedule $\OPT^m(U'_h)$. 
Insert the jobs of $K'_h$ one by one with migration: for each $j\in K'_h$, we process it in the earliest idle slots without delaying already scheduled work; if the machine is busy, pause and resume at the next idle slot (possibly on another machine).

Let $F'_j$ be the flow time of such an inserted job $j$. Then
\[
F'_j \;=\; p_j + \text{wait}_j,
\qquad
\text{wait}_j \;\le\; \frac{1}{m}\sum_{i \in U'_h \cup K'_h} p_i.
\]
By the lemma’s condition, $|K'_h|=|K^*_h|\le n^{\alpha}\sqrt{m}$, hence
\[
F(\OPT^m(U'_h \cup K'_h))
\;\le\;
F(\OPT^m(U'_h))
\;+\;
O\!\left(\frac{n^{\alpha}}{\sqrt{m}}\right)
\sum_{i \in U'_h \cup K'_h} p_i.
\]

Combining with \Cref{eqn:unknownn-unblocked2-eqn1},
\begin{align*}
\sum_{j \in U_h} F^{(2)}_j
&\le\; F(\OPT^m(U'_h \cup K'_h))
  \;+\; O\!\left(\frac{|U'_h \cup K'_h|}{n^{\alpha}\sqrt{m}}\right)\OPT \\
&\le\; F(\OPT^m(U'_h))
  \;+\; O\!\left(\frac{n^{\alpha}}{\sqrt{m}}\right)
      \sum_{i \in U'_h \cup K'_h} p_i
  \;+\; O\!\left(\frac{|S_h|}{n^{\alpha}\sqrt{m}}\right)\OPT \\
&\le\; F(\OPT(S_h))
  \;+\; O\!\left(\frac{n^{\alpha}}{\sqrt{m}}\right)
      \sum_{i \in S_h} p_i
  \;+\; O\!\left(\frac{|S_h|}{n^{\alpha}\sqrt{m}}\right)\OPT.
\end{align*}
\end{proof}

We call an interval $[t_h,t_{h+1})$ \emph{crowded} if the number of waiting blocked jobs exceeds $n^{\alpha}\sqrt{m}$ (corresponding to \Cref{lem:unknownn-crowded}), and \emph{uncrowded} otherwise (corresponding to \Cref{lem:unknownn-uncrowded}).  
Define $H_1 := \{h : [t_h,t_{h+1}) \text{ is crowded}\}$ and $H_2 := \{h : [t_h,t_{h+1}) \text{ is uncrowded}\}$. Then
\begin{align*}
   & \sum_{h \in H_1} \sum_{j \in S_h} F^{(2)}_j
     + \sum_{h \in H_2} \sum_{j \in U_h} F^{(2)}_j  \\
   &\le
   \sum_{h \in H_1 \cup H_2} \Big(
       F(\OPT(S_h))
       + O\!\left(\frac{|S_h|}{n^{1-\alpha}\sqrt{m}}\right)\OPT
       + F(\OPT(S_h)) \\
   &\qquad\qquad
       + O\!\left(\frac{n^{\alpha}}{\sqrt{m}}\right)
         \sum_{j \in S_h} p_j
       + O\!\left(\frac{|S_h|}{n^{\alpha}\sqrt{m}}\right)\OPT
   \Big) \\
   &\le
   2\,F\!\left(\OPT\!\left(\bigcup_h S_h\right)\right)
   + O\!\left(\frac{\sum_h |S_h|}{n^{1-\alpha}\sqrt{m}}\right)\OPT
   + O\!\left(\frac{n^{\alpha}}{\sqrt{m}}\right)\sum_{j \in J} p_j \\
   &\le
   O\!\left(\frac{n^{\alpha}}{\sqrt{m}}\right)\OPT.
\end{align*}

To bound $\sum_{j \in J} F^{(2)}_j$, it remains to control
\[
\sum_{h \in H_2} \sum_{j \in K_h} F^{(2)}_j.
\]

\begin{lemma}
\label{lem:unknownn_smallwaiting_blocked_uncrowded}
For the waiting time of blocked jobs during uncrowded periods,
\[
\sum_{h \in H_2} \sum_{j \in K_h} F^{(2)}_j
\;\le\; O\!\left(\frac{n^{1-\alpha}}{\sqrt{m}}\right)\OPT.
\]
\end{lemma}
\begin{proof}
During uncrowded periods, at any time $t$ the number of waiting blocked jobs is at most $n^{\alpha}\sqrt{m}$.  
Two types of intervals contribute to their waiting time.

\medskip
\noindent\textbf{Case 1.} All \SmallOnly\ machines are busy at the end. We have that the total length is at most $\frac{\sum_{j \in J} p_j}{m}$.  
Then the total waiting-time contribution is at most
\[
O\!\left(\frac{\sum_{j \in J} p_j}{m}\right)\cdot n^{\alpha}\sqrt{m}
\;\le\; O\!\left(\frac{n^{\alpha}}{\sqrt{m}}\right)\OPT.
\]

\medskip
\noindent\textbf{Case 2.} Some \SmallOnly\ machines are not busy at the end.  
This only happens at a type-A kill event.  
If a blocked job is waiting, it is in its last small period, so it is small and all machines are busy at that moment.  
Hence, if some \SmallOnly\ machines are idle at the end, a type-A kill event must have occurred when $N$ was updated to $2N$ --- the only time a small job can become large.  
At such an update, at most $N$ jobs have been released, so there are at most $N$ blocked jobs.  
Each killed job was small just before the update, hence its size is at most $\OPT/(N^{\alpha}\sqrt{m})$.  
Thus each such event contributes to $\sum_{h \in H_2} \sum_{j \in K_h} F^{(2)}_j$ at most
\[
O\!\left(\frac{\OPT \cdot N^{1-\alpha}}{\sqrt{m}}\right) . 
\]

As $N$ doubles through $1,2,4,\ldots$ up to $n<2^k\le 2n$, the total contribution sums to
\[
O\!\left(\frac{n^{1-\alpha}}{\sqrt{m}}\right)\OPT.
\]
\end{proof}

Combining all the lemmas above, we can prove the main theorem.

\begin{theorem}
\label{thm:unknown_det_restart}
    \Cref{alg:multi-deterministic-restart-m-unknownn} is an online polynomial-time deterministic algorithm with kill-and-restart for total flow time minimization that is $O(n^{\alpha}/\sqrt{m})$-competitive against the preemptive offline solution without prior knowledge of $n$, where $\alpha = (\sqrt{5}-1)/2$, provided that $m\geq 2$.
\end{theorem}
\begin{proof}
By \Cref{lem:unknownn_large_waiting},
\[
\sum_{j\in J} F_j^{(1)} \;\le\; O\!\left(\frac{n^{\alpha}}{\sqrt{m}}\right)\OPT.
\]
By \Cref{lem:unknownn-crowded,lem:unknownn-uncrowded,lem:unknownn_smallwaiting_blocked_uncrowded},
\[
\sum_{j\in J} F_j^{(2)} \;\le\; O\!\left(\frac{n^{\alpha}}{\sqrt{m}}\right)\OPT.
\]
Combining,
\[
\sum_{j\in J} F_j
\;=\; \sum_{j\in J}\!\bigl(F_j^{(1)}+F_j^{(2)}+p_j\bigr)
\;\le\; O\!\left(\frac{n^{\alpha}}{\sqrt{m}}\right)\OPT.
\]
\end{proof}

\newpage
\section{Lower Bounds for Algorithms with the Knowledge of \texorpdfstring{$n$}{n}}
\label{sec:lower}

In this section, we present our lower bound results for algorithms that have prior knowledge of $n$. Specifically, for every fixed pair $(n,m)$, we construct an adversarial instance that establishes a lower bound. The results are organized into two parts.

The first part concerns randomized algorithms. We construct hard instances via distributions, following Yao's Minimax Principle. We begin with the single-machine setting as a warm-up, then generalize this construction to multiple machines, and finally extend it to the setting that allows kill-and-restart.

In the second part, we turn to deterministic algorithms. We first establish an $\Omega(n/\log n)$ lower bound for deterministic algorithms with kill-and-restart on a single machine. Finally, we present a lower bound of $\Omega(n/m^2 + \sqrt{n/m})$ for deterministic non-preemptive algorithms.

\subsection{Warm-Up: Randomized Lower Bound on a Single Machine}
\label{sec:single-randomized-lb}
We start with the lower bound that any randomized non-preemptive online algorithm for a single machine is $\Omega(\sqrt n)$ as a warm-up. The counterexample is constructed as \Cref{alg:single-randomized-lb}, which is a distribution of two possibilities: the only difference is whether the first $k$ $\varepsilon$-jobs are released at time $1$ or $2$.

\begin{algorithm}[H]
    \caption{Counterexample Construction for $\Omega(\sqrt{n})$ Lower Bound}
    \label{alg:single-randomized-lb}
    \SetAlgoLined
    \KwIn{Number of jobs $n$}
    $k \gets \lfloor\sqrt{n - 2}\rfloor, \varepsilon \gets 0.5 n ^ {-2}$\;
    Release job $1$ at time $0$ with processing time $2$ \;
    Flip an independent fair coin: $r \sim \text{Bernoulli}(1/2)$ \;
    Release $k$ jobs at time $1 + r$ with processing time $\varepsilon$ \tcp*{released at $1$ or $2$}
    \For{$t = 2$ \KwTo $k$}{
        Release $k$ jobs at time $1 + t$ with processing time $\varepsilon$ \tcp*{released at $3, \dots, k + 1$}
    }
\end{algorithm}

\begin{lemma}
    \label{thm:lb-single-random}
    The competitive ratio in flow time of any randomized non-preemptive online kill-and-restart algorithm for a single machine is $\Omega(\sqrt{n})$.
\end{lemma}

\begin{proof}
    We analyze the counterexample described in \Cref{alg:single-randomized-lb}. For each of the two possibilities, the optimal algorithm can arrange job $1$ at time $0$ or $1$ correctly and delay each $\varepsilon$-job by at most $O(k^2 \varepsilon)$. Therefore, the optimal flow time is always $O(1)$. However, for any deterministic online algorithm, it has to decide whether to start job $1$ at the beginning or not:
        \begin{itemize}
            \item If the algorithm starts job $1$ at time $0$, then with probability $1/2$, the algorithm will delay the first $k$ $\varepsilon$-jobs by $1$ for the case the randomized $\varepsilon$-jobs are released at time $1$;
            \item Otherwise, the algorithm cannot start job $1$ before time $k + 1$ for the case the randomized $\varepsilon$-jobs are released at time $2$, or it will delay at least $k$ $\varepsilon$-jobs by $1$.
        \end{itemize}
        Therefore, the optimal deterministic algorithm has an expected flow time of $\Omega(\sqrt n)$ against the adversary, and hence the randomized algorithm has the lower bound $\Omega(\sqrt n)$ by Yao's minimax theorem.
\end{proof}


\subsection{Generalization to Multiple Machines}
\label{sec:multi-lb}

The main idea is to summarize the single-machine hard instance into a \emph{busy gadget} and then replicate it $m$ times to form a \emph{busy batch}. To extend this construction to multiple machines, we must address key differences related to small jobs. In particular, small jobs must have significant (non-negligible) processing time; otherwise, a deterministic algorithm could process a batch of small jobs --- originally intended for $m$ machines --- on a single machine with little loss. Since a batch of small jobs increases $\OPT$ by only a constant factor, a single large job is not sufficient to create a gap in the competitive ratio. Therefore, we release one large job in each busy gadget. Let $k = \Theta(\sqrt{n/m})$ be the target competitive ratio. We release $k$ batches of $m$ gadgets each. Each gadget, parameterized by a start time $t$, contains $2k + 1$ jobs structured as follows and is illustrated in \Cref{fig:multi-construction}:

\begin{figure}[t]
    \centering
    \includegraphics[width=0.8\linewidth]{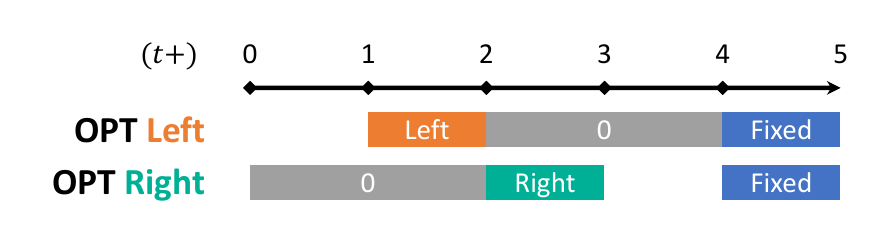}
        \begin{tabular}{p{0.25\linewidth} p{0.2\linewidth} p{0.2\linewidth} p{0.05\linewidth}}
        \toprule
        \textbf{Job Type} & \textbf{Indices ($i$)} & \textbf{$r_i$} & \textbf{$p_i$} \\
        \midrule
        \textcolor{Gray0}{\rule{0.25cm}{0.25cm}} Large 0  & $0$ & $0$ & $2$ \\
        \addlinespace
        \textcolor{FixedBlue}{\rule{0.25cm}{0.25cm}} Fixed          & $\{1, \dots, k\}$ & $4 + (i-1)/{k}$ & $1/k$ \\
        \textcolor{LeftOrange}{\rule{0.25cm}{0.25cm}} Left (when $c = 0$)    & $\{k+2, \dots, 2k+1\}$ & $1 + (i-(k+2))/{k}$ & $1/k$ \\
        \textcolor{RightTeal}{\rule{0.25cm}{0.25cm}} Right (when $c = 1$)   & $\{k+2, \dots, 2k+1\}$ & $2 + (i-(k+2))/{k}$ & $1/k$ \\
        \bottomrule
    \end{tabular}
    \caption{A gadget of $2k + 1$ jobs defined in \Cref{alg:multi-lb}, with two random options. }
    \label{fig:multi-construction}
\end{figure}

\begin{itemize}
    \item \textbf{A Large Job}: Job $0$, with processing time $2$ and release time $t$. 
    \item \textbf{Fixed Small Jobs}: $k$ jobs, each with processing time $1/k$, released sequentially in the interval $[t+4, t+5)$.
    \item \textbf{Random Small Jobs}: $k$ jobs, each with processing time $1/k$, released sequentially in either $[t+1, t+2)$ (if a random coin $c=0$) or $[t+2, t+3)$ (if $c=1$).
\end{itemize}

Finally, we provide a detailed description of how the busy gadget is replicated across all $m$ machines and repeated $k$ times over time in \Cref{alg:multi-lb}. A notable aspect is that all $m$ gadgets forming a batch $b$ (which starts at time $t=5b$) share the same random coin flip $c_b$.

\begin{algorithm}[t]
    \caption{Counterexample Construction for $\Omega(\sqrt{n / m})$ Lower Bound}
    \label{alg:multi-lb}
    \SetAlgoLined
    \KwIn{Number of jobs $n$, number of machines $m$}
    $k \gets \lfloor\sqrt{n / (2m) + 1/16} - 1/4\rfloor$ \tcp*{$k = \Theta(\sqrt{n / m})$ and $km \cdot (2k + 1) \leq n$}
    \BlankLine
    \SetKwFunction{FGadget}{Gadget}
    \SetKwProg{Fn}{Function}{:}{}
    \Fn{\FGadget{$t, c$}}{
        $\mathcal{J} \gets \varnothing$ \tcp*{Gadget job (multi-)set}
        $\mathcal{J} \gets \mathcal{J} \cup \{(r_0 = t, p_0 = 2)\}$\;
        $\mathcal{J} \gets \mathcal{J} \cup \{(r_i = t + 4 + (i - 1) / k, p_i = 1/k) : i = 1 \text{~to~} k\}$\;
        $\mathcal{J} \gets \mathcal{J} \cup \{(r_{k + i} = t + 1 + c + (i - 1) / k, p_{k + i} = 1 / k)  : i = 1 \text{~to~} k\}$\;
        \KwRet{$\mathcal{J}$}\;
    }
    \BlankLine
    \For{$b = 0$ \KwTo $k - 1$}{
        Flip an independent fair coin: $c_b \sim \text{Bernoulli}(1/2)$ \tcp*{$0$: Left, $1$: Right}
        Release the $b$-th batch consisting of $m$ copies of $\FGadget(5b, c_b)$ \; 
    }
\end{algorithm}

\begin{lemma}
    \label{fact:multi-lb-opt}
    $\OPT \leq 6mk$ for the instance constructed in \Cref{alg:multi-lb}.
\end{lemma}
\begin{proof}
    In each busy gadget, we can schedule the large job either at $5b+0$ or $5b+2$ depending on $c$, and every small job upon arrival. As a result, the total flow time in this gadget is at most $4$ (the large job) $+1$ (fixed small jobs) $+1$ (random small jobs). Therefore $\OPT \leq 6mk$ by summing up all gadgets. 
\end{proof}

Moving to the analysis of an arbitrary online algorithm, we aim to prove that the algorithm's flow time is at least $\Omega(mk^2)$. First, we will show that when a large job is placed such that it conflicts with the small jobs (having an intersection of at least $0.5$), it significantly increases the flow time of those small jobs.

\begin{Definition}
    In a given batch, the set of $mk$ jobs released in a time interval $[t, t + 1)$ constitutes a \textbf{small job period}. A large job \textbf{conflicts with} this period if its execution starts within the interval $[t - 1.5, t + 0.5]$. This definition guarantees that the large job's execution (of length 2) overlaps with the period $[t, t+1)$ for a duration of at least $0.5$.
\end{Definition}

\begin{lemma}
    \label{lem:multi-lb-small}
    If $x$ large jobs conflicts with a small job period at $[t,t+1)$, then the total flow time of these small jobs within $[t, t+1)$ is at least $xk / 8$. 
\end{lemma}

\begin{proof}

Define $f(u) = m \cdot \lceil ku \rceil$ as the number of small jobs released at time $t+u$, and let $g(u)$ be the number of machines not processing large jobs at time $u$. The number of small jobs that can be completed by time $u$ is at most $\lfloor k \int_0^u g(y) dy \rfloor$. 
The total flow time of the $km$ small jobs in the batch is no less than 
\begin{align*}
    F &\geq \int_0^1 f(u) - \lfloor k \int_0^u g(y) \mathrm{d}y \rfloor \,\mathrm{d}u \\
    &\geq k \cdot \int_0^1 mu - \int_0^u g(y) \mathrm{d}y \,\mathrm{d}u \\
    &= k \cdot \int_0^1 (m - g(u)) \cdot (1 - u) \,\mathrm{d}u,
\end{align*}

If no large jobs conflict with this period, $g(u)$ is defined as $m$ for every $u \in [0,1)$. By the condition of the lemma, we have $x$ jobs that decrease $g(u)$ by $1$ with a duration of at least $0.5$. The minimum contribution of the lower bound occurs when each large job conflicts the period at $[0.5,1)$, yielding 

\begin{align*}
    F &\geq x k \cdot \min_{s \in [0,0.5]} \int_s^{s+0.5} 1 \cdot (1-u) \,\mathrm{d}u\\
    &= xk \cdot \int_{0.5}^{1} (1-u) \,\mathrm{d}u = xk / 8.
\end{align*}
\end{proof}

\Cref{lem:multi-lb-small} establishes that an incorrectly placed large job incurs a total flow time of $\Omega(k)$. In the next lemma, we aim to prove that a deterministic algorithm must either suffer a large flow time because of incorrectly placed large jobs or defer the large jobs to the next gadget. We capture the total flow time using a potential function $\Phi(b,d)$ that represents a lower bound on the expected total flow time in the time interval $[5b - 0.5,\infty)$, conditioned on the algorithm having $d$ unscheduled large jobs at time $5b - 0.5$.

\begin{theorem}
    \label{thm:multi-lb}
   Any randomized online algorithm for total flow time minimization on parallel machines has a competitive ratio of at least $\Omega(\sqrt{n/m})$, for every fixed pair $(n, m)$.
\end{theorem}

\begin{proof}

For a fixed busy batch $b$, we analyze any fixed deterministic algorithm's behavior on large jobs during the interval $[5b - 0.5,\; 5(b+1) - 0.5)$.\footnote{Note that when $b = 0$, the negative part doesn't exist, and it will not affect the discussion when the negative part come into consideration.} Since $b$ is fixed in context, we simplify the interval to $[-0.5,\; 4.5)$.

For a fixed arbitrary realization of $c_0, \dots, c_{b-1}$, consider two possible realizations of $c_b$. Let $N_b$ be the random variable (w.r.t. $c_b$) representing the number of large jobs the algorithm starts within $[-0.5,\; 4.5)$, and $I_b$ be the random variable (w.r.t. $c_b$) representing the number of large jobs scheduled in this batch that conflict with small job periods.

Note that before time $1$, the algorithm does not know the realization of $c_b$ this round. Therefore, its behavior before time $1$ is not related to the result of $c_b$. 

\begin{itemize}
    \item Let $x$ be the number of large jobs started in $[-0.5,\; 1)$. If small jobs are Left, all these $x$ jobs contribute to $I$.
\end{itemize}

After time $1$, the algorithm's scheduling decisions may depend on $c_b$. We consider the two cases separately:

\begin{itemize}
    \item {Right case (w.p. $0.5$):} Small jobs occupy $[2,\; 3)$ and $[4,\; 5)$. Any large job started in $[1,\; 4.5)$ must conflict with a small job. Let $y$ be the number of such large jobs, which will contribute $y$ to $I_b$.

    \item {Left case (w.p. $0.5$):} Small jobs occupy $[1,\; 2)$ and $[4,\; 5)$. The algorithm can safely schedule up to $m$ large jobs in the gap $[1.5,\; 2.5)$ without conflicting any small job period. Let $z$ be the number of large jobs scheduled in $[1,\; 1.5), [2.5, \; 4.5)$; then at least $z$ jobs conflict with small jobs.
\end{itemize}

Combining all cases, we compute the expectations:
\[
\E_{c_b}(N_b) \leq x + \frac{y}{2} + \frac{z + m}{2}, \qquad
\E_{c_b}(I_b) \geq \frac{x}{2} + \frac{y}{2} + \frac{z}{2}.
\]

Let $F^{(S)}_b$ denote the total flow time of small jobs in batch $b$. By \Cref{lem:multi-lb-small},
\[
\E_{c_b}(F^{(S)}_b) \geq \frac{k}{8} \cdot \E_{c_b}(I_b) \geq \frac{k}{8} \cdot \left(\frac{\E_{c_b}(N_b) - \frac{m}{2}}{2}\right) \geq \frac{k}{16} \cdot \left( \E_{c_b}(N_b) - \frac{m}{2}\right).
\]
The above argument applies to each $b$ with all possible $c_0, \dots, c_{b-1}$, by summing them up we have
\[
\E \left(F^{(S)}\right) \geq \sum_{b=0}^{k-1} \E\left(F^{(S)}_b\right) \geq \frac{k}{16} \cdot \sum_{b=0}^{k-1}\left(  \E(N_b) - \frac{m}{2} \right).
\]

For the total flow time $F^{(L)}$ incurred by large jobs, we count the number of unfinished and newly released large jobs at time $0$ for each batch. By the fact that unfinished jobs will be completed later than $t + 1$ for a batch starting at $t$, they must contribute to the total flow time in $[t, t + 1)$. Therefore,
\begin{align*}
\E(F^{(L)}) &\geq \sum_{i = 0}^{k - 1} \left((i + 1)\cdot m - \sum_{b = 0}^{i - 1} \E(N_b)\right)
\geq \E\left( \sum_{i=0}^{k-1} \sum_{b=0}^{i} \left( m - \E(N_b) \right) \right) \\
&= \sum_{b=0}^{k-1} \left( m - \E(N_b) \right) \cdot (k - b) \\
&\geq \frac{k}{16} \cdot \sum_{b=0}^{\lfloor 15k/16 \rfloor} \left( m - \E(N_b) \right).
\end{align*}

Combining the two bounds, we have
\[
\E(F^{(L)}) + \E(F^{(S)}) \geq \frac{k}{16} \cdot \sum_{b=0}^{\lfloor 15k/16 \rfloor} \frac m 2 = \Omega(k^2 m).
\]
By \Cref{fact:multi-lb-opt}, $\OPT = O(mk)$, then any deterministic algorithm is $\Omega(k)=\Omega(\sqrt{n/m})$ competitive on the counterexample, and we can conclude the theorem by Yao's minimax principle.
\end{proof}

\subsection{Generalization to Multiple Machines with Kill-and-Restart}
\label{sec:multi-randomized-lb-restart}


The details of the construction are presented in \Cref{fig:multi-construction-restart} and \Cref{alg:multi-restart-lb}. The overall structure of this hard instance follows that of \Cref{sec:multi-lb}, with one key difference: Each gadget now contains two large jobs of different sizes. This modification strengthens the lower bound instance, as algorithms must now not only decide whether to schedule a large job, but also determine which type of large job to schedule. Intuitively, in the previous hard instance, the algorithm could start a large job early and decide whether to kill it after the randomness is revealed. However, in the current construction, the algorithm must decide which large job to start in advance; a wrong decision can lead to a substantial cost.



\begin{figure}[ht]
    \centering
    \includegraphics[width=1.0\linewidth]{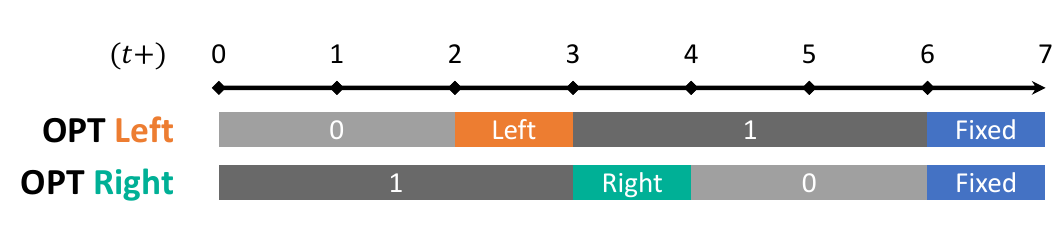}
        \begin{tabular}{p{0.25\linewidth} p{0.2\linewidth} p{0.2\linewidth} p{0.05\linewidth}}
        \toprule
        \textbf{Job Type} & \textbf{Indices ($i$)} & \textbf{$r_i$} & \textbf{$p_i$} \\
        \midrule
        \textcolor{Gray0}{\rule{0.25cm}{0.25cm}} Large 0  & $0$ & $0$ & $2$ \\
        \textcolor{Gray1}{\rule{0.25cm}{0.25cm}} Large 1  & $1$ & $0$ & $3$ \\
        \addlinespace 
        \textcolor{FixedBlue}{\rule{0.25cm}{0.25cm}} Fixed          & $\{2, \dots, k+1\}$ & $6 + (i-2)/{k}$ & $1/k$ \\
        \textcolor{LeftOrange}{\rule{0.25cm}{0.25cm}} Left (when $c = 0$)    & $\{k+2, \dots, 2k+1\}$ & $2 + (i-(k + 2))/{k}$ & $1/k$ \\
        \textcolor{RightTeal}{\rule{0.25cm}{0.25cm}} Right (when $c = 1$)   & $\{k+2, \dots, 2k+1\}$ & $3 + (i-(k + 2))/{k}$ & $1/k$ \\
        \bottomrule
    \end{tabular}
    \caption{A gadget of $2k + 2$ jobs defined in \Cref{alg:multi-restart-lb}, with two random options. }
    \label{fig:multi-construction-restart}
\end{figure}

\begin{algorithm}[t]
    \caption{Counterexample Construction for $\Omega(\sqrt{n / m})$ Lower Bound}
    \label{alg:multi-restart-lb}
    \SetAlgoLined
    \KwIn{Number of jobs $n$, number of machines $m$}
    $k \gets \lfloor\sqrt{n / (2m) + 1/4} - 1/2\rfloor$ \tcp*{$k = \Theta(\sqrt{n / m})$ and $km \cdot (2k + 2) \leq n$}
    \BlankLine
    \SetKwFunction{FGadget}{Gadget}
    \SetKwProg{Fn}{Function}{:}{}
    \Fn{\FGadget{$t, c$}}{
        $\mathcal{J} \gets \varnothing$ \tcp*{Gadget job (multi-)set}
        $\mathcal{J} \gets \mathcal{J} \cup \{(r_0 = t, p_0 = 2),\ (r_1=t,p_1=3)\}$\;
        $\mathcal{J} \gets \mathcal{J} \cup \{(r_i = t + 6 + (i - 2) / k, p_i = 1/k) : i = 2 \text{~to~} k+1\}$\;
        $\mathcal{J} \gets \mathcal{J} \cup \{(r_{k + i} = t + 2 + c + (i - 2) / k, p_{k + i} = 1 / k)  : i = 2 \text{~to~} k+1\}$\;
        \KwRet{$\mathcal{J}$}\;
    }
    \BlankLine
    \For{$b = 0$ \KwTo $k - 1$}{
        Flip an independent fair coin: $c_b \sim \text{Bernoulli}(1/2)$ \tcp*{0: Left, 1: Right}
        Release the $b$-th batch consisting of $m$ copies of $\textsc{Gadget}(7b, c_b)$ \;
    }
\end{algorithm}

\begin{lemma}
    \label{fact:multi-restart-lb-opt}
    $\OPT \leq 11mk$ for the instance constructed in \Cref{alg:multi-restart-lb}.
\end{lemma}
\begin{proof}
    In each busy gadget, if $c=0$, then schedule large job $0$ at $0$ and large job $1$ at $3$; if $c=1$, then schedule large job $1$ at $0$ and large job $0$ at $4$. In this way, no large jobs intersect with any small jobs. As a result, the total flow time in this gadget is at most $9$ (large jobs) $+1$ (fixed small jobs) $+1$ (random small jobs). Therefore, $\OPT \leq 11mk$ by summing up all gadgets.
\end{proof}



\begin{theorem}
    \label{thm:multi-lb-restart}
   Any randomized online algorithm for total flow time minimization on parallel machines has a competitive ratio of at least $\Omega(\sqrt{n/m})$, for every fixed pair $(n, m)$, even with the capability of kill-and-restart.
\end{theorem}

\begin{proof}
For a fixed busy batch $b$, we analyze any fixed deterministic algorithm's behavior on large jobs during the interval $[7b - 0.5,\; 7(b+1) - 0.5)$. Since $b$ is fixed in context, we simplify the interval to $[-0.5,\; 6.5)$.

For a fixed realization of $c_0, \dots, c_{b-1}$, consider two possible realizations of $c_b$. Let $N_{b, 2}, N_{b, 3}$ be the random variable (w.r.t. $c_b$) representing the number of jobs with processing time $2, 3$ the algorithm start within $[-0.5,\; 6.5)$ and don't restart, and $I_{b,2}, I_{b,3}$ defined similarly that intersect with small job periods. Let $N_b = N_{b, 2} + N_{b, 3}, I_b = I_{b, 2} + I_{b, 3}$.

Note that before time $2$, the algorithm does not know the realization of $c_b$ this round. Therefore, its behavior before time $2$ is not related to the result of $c_b$. 

\begin{itemize}
    \item Let $x_2$ be the number of jobs with processing time $2$ started in $[-0.5,\; 2)$.
    \item Let $x_3$ be the number of jobs with processing time $3$ started in $[-0.5,\; 2)$.
\end{itemize}

After time $2$, the algorithm scheduling decisions may depend on $c_b$. We consider the two cases separately:
\begin{itemize}
    \item {Left case (w.p. $0.5$):} None of the $x_3$ jobs can finish without intersection, and at most $x_{2}^{(L)} =\min \{m, x_2\}$ out of the $x_2$ jobs can finish without intersection. Let $y_2^{(L)}$ be the number of jobs with processing time $2$ starting in $[2, 6.5)$ that finish without intersection, and $y_3^{(L)}$ be the number for processing time $3$. At most $m$ jobs out of $y_2^{(L)} + y_3^{(L)}$ can finish without intersection.
    \item {Right case (w.p. $0.5$):} At most $m$ out of these $x_2 + x_3$ jobs could end up without intersection with small jobs; let $x_2^{(R)}$ and $x_3^{(R)}$ be the number of such jobs, respectively. At most $m$ jobs with processing time $2$ released in $[2, 6.5)$ can finish without intersection; let $y_2^{(R)}$ be the number of such jobs. None of the jobs with processing time $3$ can finish without intersection.
\end{itemize}

Then, 
\begin{align}
\E_{c_b}(N_{b} - I_{b}) &\leq \frac 1 2 \left(x_2^{(L)} + y_2^{(L)} + y_3^{(L)}\right) + \frac 1 2 \left(x_2^{(R)} + x_3^{(R)} + y_2^{(R)}\right)\\
&\leq \frac 1 2 \left(x_2^{(L)} + m\right) + \frac 1 2 \left(x_2^{(R)} + x_3^{(R)} + y_2^{(R)}\right)\\
&\leq \frac 3 2 m + \frac 1 2 x_2^{(L)} \leq \frac 3 2 m + \frac 1 2 x_2
\label{eq:1}
\end{align}
and
\begin{align} 
\E_{c_b}(N_{b,3} - I_{b,3}) &\leq \frac 1 2 y_3^{(L)} + \frac 1 2 x_3^{(R)} \leq \frac 1 2 m + \frac 1 2 \left(m - x_2\right) \leq m - \frac 1 2 x_2. 
\label{eq:2}
\end{align}

By $(\ref{eq:1}) + (\ref{eq:2})$, we have
\begin{align*}
\E_{c_b}(N_b - I_b) + \E_{c_b}(N_{b,3} -I_{b,3}) 
    &\leq \frac 3 2 m + \frac 1 2 x_2 +  m - \frac 1 2 x_2 = \frac 5 2 m,
\end{align*}
and hence $\E(N_b - I_b) + \E(N_{b,3} - I_{b,3}) \leq \frac 5 2 m$ since it doesn't rely on the specific values of $c_0,\ldots,c_{b-1}$.

Consider the flow time by the large jobs. If we consider the jobs with processing time $2$ or $3$ and their intersections only, we have
\begin{align*}
\E(F^{(2,3)}) &\geq \sum_{b = 0}^{k - 1} \sum_{i = 0}^{b} (2m - E(N_b)) + \sum_{b = 0}^{k - 1} E(I_b) \cdot k / 8 \\
&\geq  \sum_{b = 0}^{k - 1} (2m - E(N_b)) \cdot (b - k) + \sum_{b = 0}^{k - 1} E(I_b) \cdot k / 8 \\
&\geq \sum_{b = 0}^{\lfloor \frac{7}{8} k\rfloor} \left(2m - E(N_b - I_b)\right) \cdot k/8.
\end{align*}
Similarly, consider the jobs with processing time $3$ and their intersections only, we have
\begin{align*}
\E(F^{(3)}) &\geq \sum_{b = 0}^{\lfloor \frac{7}{8} k\rfloor} (m - \E(N_{b,3} - I_{b,3})) \cdot k/8.
\end{align*}

Therefore, the total expectation of total flow time is at least
\begin{align*}
\E(F) = \frac 1 2 \cdot \left(2 \cdot \E(F)\right)
& \geq \frac 1 2 \cdot \left(\E(F^{(2, 3)}) + \E(F^{(3)}) \right) \\
& \geq \frac k {16} \cdot \left( \sum_{b = 0}^{\lfloor \frac{7}{8} k\rfloor} 3m - \left(\E(N_b - I_b) + \E(N_{b,3} -I_{b,3})\right)\right) \\
& \geq \frac k {16} \cdot \left( \sum_{b = 0}^{\lfloor \frac{7}{8} k\rfloor} m/2\right) = \Omega(k^ 2 m).
\end{align*}

By \Cref{fact:multi-restart-lb-opt}, $\OPT = O(k m)$, hence any deterministic algorithm is $\Omega(k) = \Omega(\sqrt{n / m})$ competitive and the theorem is proved.
\end{proof}


\subsection{Deterministic Algorithms: Single Machine with Kill-and-Restart}
\label{sec:single-restart-lb}
In this section, we show that the competitive ratio of any deterministic online algorithm with kill-and-restart is at least $\Omega(n / \log n)$ for the single-machine case. The adversarial instance, detailed in \Cref{alg:single-restart-lb}, uses a two-phase construction.

\begin{itemize} 
    \item \textbf{In Phase 1,} the adversary uses two large jobs and carefully timed small jobs to force the online algorithm, $\ALG$, to end the phase with at least one large job still unstarted. If $\ALG$ manages to complete both large jobs, it is forced to incur an $\Omega(n)$ flow time, ending the game.
    \item \textbf{In Phase 2,} the adversary exploits the presence of the unstarted large job. It releases a stream of tiny $\varepsilon$-jobs at logarithmically increasing intervals before each restart of the large job. As we will show, this construction forces $\ALG$ to incur a total flow time of $\Omega(n / \log n)$, regardless of its specific kill-and-restart strategy.
\end{itemize}

\begin{algorithm}[H]
    \caption{Adversary for Deterministic Lower Bound $\Omega(n / \log n)$}
    \label{alg:single-restart-lb}
    \SetAlgoLined
    \KwIn{Deterministic online algorithm $\ALG$, number of jobs $n$}
    $c \gets 0.5, \varepsilon \gets 0.5 n ^ {-2}$\;
    \SetKwBlock{PhaseOne}{\textbf{Phase 1: Construct an unsolved job}}{}
    \SetKwBlock{PhaseTwo}{\textbf{Phase 2: Exploit the unsolved job $u$}}{}
    \PhaseOne{
        Release job $1$ at $0$ with processing time $4$ \;
        Release job $2$ at $2$ with processing time $1$ \;
        \uIf{
            $\ALG$ is working on job $1$ before $3$
        }{
            Release $c/2 \cdot n$ jobs at $3$ with processing time $\varepsilon$ \;
            Release $c/2 \cdot n$ jobs at $7$ with processing time $\varepsilon$ \;
            $t^* \gets 7$\;
        }
        \Else{
            Release $c \cdot n$ jobs at $5$ with processing time $\varepsilon$ \; 
            $t^* \gets 5$\;
        }
        Wait until $\ALG$ starts any job released at $t^*$, let $t$ be the current time \;
        \uIf{
            jobs $1, 2$ are both completed or $t - t^* \geq 1$
        }{
            Release $n - 2 - cn$ jobs at time $t$ with processing time $\varepsilon$ \tcp*{Ensure $n$ jobs}
            \Return\tcp*{$\ALG = \Omega(n), \OPT = O(1)$}
        }
        \Else{
            Wait $c\cdot n\cdot \varepsilon$ units of time \tcp*{Ensure $\OPT$ finishes all jobs}
            \textbf{continue} to \textbf{Phase 2} with any unsolved job $u$ among jobs $1, 2$ \;
        }
    }
    \PhaseTwo {
        $T = \{ H_i / H_n : i = 1 \dots n, H_i = \sum_{k = 1}^{i} 1/k\}$ \tcp*{Set of thresholds}
        $l \gets n - 2 - cn$ \tcp*{Number of jobs left}
        \While{$l > 0$}{
            Wait until $\ALG$ is working on job $u$, let $t$ be the current time \;
            Let $t'$ be the time when $\ALG$ starts job $u$ \;
            $\tau \gets \min \{ x \in T \mid x > t - t' \}$ \;
            \If{
                $\ALG$ hasn't killed job $u$ before time $t' + \tau$
            }{
                Release a job at time $t' + \tau$ with processing time $\varepsilon$ \;
                $l \gets l - 1$ \;
            }
        }
    }
\end{algorithm}

\begin{lemma}
\label{lem:restart-phase-1}
Against the adversary in Phase $1$ of \Cref{alg:single-restart-lb}, any deterministic online algorithm either incurs a flow time of $\Omega(n)$ or proceeds to Phase 2 with at least one of job 1 or job 2 unstarted. In all cases, $\OPT$ remains $O(1)$ and all jobs are solved.
\end{lemma}

\begin{proof}
    
See \Cref{fig:phase1} which enumerates all possible strategies for $\ALG$ that (1) only make new decisions at new job arrivals, and (2) do not let the machine idle since $\ALG$ has the ability to restart. Though it only enumerates these well-formed algorithmic decisions and does not detail other strategies that do not satisfy these conditions, we argue that the description of the adaptive adversary is sufficiently valid for them.

\begin{figure}[ht]
    \includegraphics[width=\linewidth]{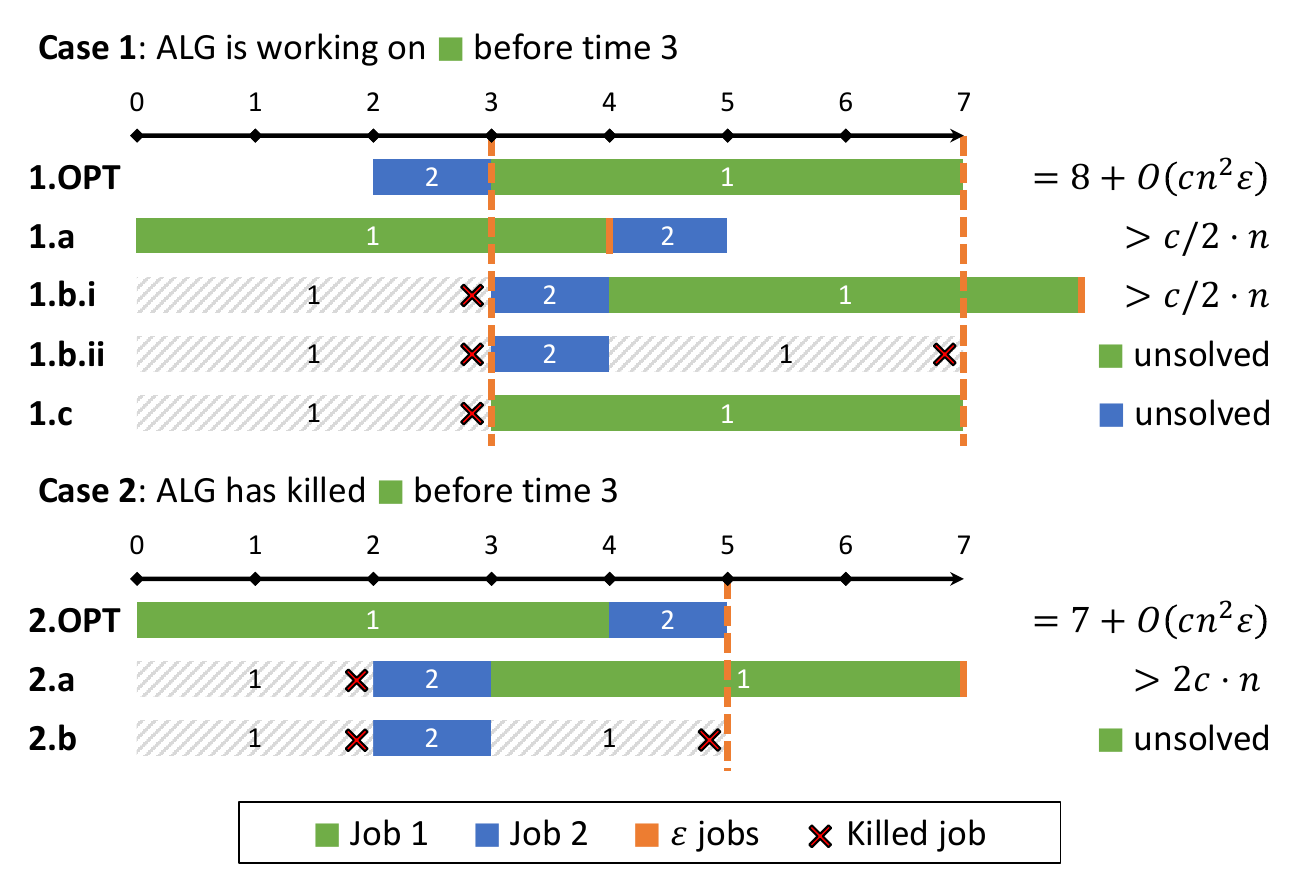}
    \caption{The cases in Phase 1 for $\ALG$.}
    \label{fig:phase1}
\end{figure}
\begin{enumerate}
    \item If $\ALG$ is working on job $1$ at time $3$, then $\OPT$ could solve job $2$ by $3$ and job $1$ by $7 + c/2 \cdot n \varepsilon$. For $\ALG$, however:
    \begin{enumerate}
        \item If $\ALG$ keeps working on job $1$ until finished, the flow time is at least $c / 2 \cdot n$;
        \item If $\ALG$ kills at time $3$ to solve $\varepsilon$-jobs, switch to job $2$ and then switch to job $1$:
        \begin{enumerate}
            \item If it doesn't switch to $\varepsilon$-jobs at time $7$, the flow time is at least $c / 2 \cdot n$;
            \item Otherwise $\ALG$ must switch to $\varepsilon$-jobs at time $7$, leaving job $1$ unsolved.
        \end{enumerate}
    \end{enumerate}
    \item Otherwise, $\ALG$ cannot finish job $1, 2$ by time $5$ while $\OPT$ can. In this case $\ALG$ would switch to job $2$ at time $2$ and then to solve job $1$, or it would fall back to case 1.
    \begin{enumerate}
        \item If $\ALG$ keeps working on job $1$ until finished, the flow time is at least $2c \cdot n$;
        \item Otherwise $\ALG$ must switch to $\varepsilon$-jobs at time $5$, leaving job $1$ unsolved.
    \end{enumerate}
\end{enumerate}
\end{proof}

\begin{lemma}
\label{lem:restart-phase-2}
In Phase 2 of \Cref{alg:single-restart-lb}, an online algorithm $\ALG$ must schedule the jobs released in this phase along with an unstarted large job $u$ (with $p_u \ge 1$) carried over from Phase 1. The total flow time for $\ALG$ on the Phase 2 jobs is at least $\Omega(n / \log n)$. In contrast, an optimal offline solution for an instance containing only the Phase 2 jobs (with no initial unstarted job) has a flow time of $O(1)$.
\end{lemma}

\begin{proof}

Let $N = n - 2 - cn$ be the total number of $\varepsilon$-jobs released in this phase. Consider any deterministic strategy employed by $\ALG$. This strategy will result in some number of restarts of the large job $u$, say $r \ge 0$. These restarts partition the $N$ small jobs into $r+1$ groups. Let $a_i$ be the number of $\varepsilon$-jobs released during the $i$-th run of job $u$ (for $i=1, \dots, r+1$), where a \emph{run} is a period of continuous work on $u$. We have $\sum_{i=1}^{r+1} a_i = N$.

The key insight is that the total flow time of the small jobs is lower-bounded by the same quantity, regardless of how the $a_i$ are partitioned. By construction of the adversary, we can count total flow time before the next kill (or completion when $i = r + 1$) of the large job by summing up the flow time of $j$-th job in the $i$-th run, where we use the $j = 0$ to describe the large job and $j = 1, \dots, a_i$ to describe the $\varepsilon$-jobs in release order:
\begin{align*}
    \sum_{i = 1}^{r + 1} \sum_{j = 0}^{a_i} (H_{a_i} - H_{j}) / H_n
    &= \sum_{i = 1}^{r + 1} \sum_{j = 0}^{a_i} \left(\sum_{k=j + 1}^{a_i} \frac{1}{k}\right) / H_n \\
    &= \sum_{i = 1}^{r + 1} \sum_{k = 1}^{a_i} \frac{1}{k} \sum_{j = 0}^{k - 1} 1 / H_n \\
    &= \sum_{i = 1}^{r + 1} \frac{a_i} {H_n} = \frac{N}{H_n}. 
\end{align*}
Since $N=\Theta(n)$ and $H_n = \Theta(\log n)$, the flow time is $\Omega(n / \log n)$.
In contrast, an offline $\OPT$ can serve every $\varepsilon$-job upon its arrival, leading to a total flow time of $N \cdot \varepsilon = O(1)$.
\end{proof}

\begin{theorem}
    \label{thm:restart-lb}
    Any deterministic online algorithm with kill-and-restart for total flow time minimization on a single machine has a competitive ratio of at least $\Omega(n / \log n)$.
\end{theorem}
\begin{proof}
    From \Cref{lem:restart-phase-1} and \Cref{lem:restart-phase-2}, the flow time of any deterministic algorithm is at least $\Omega(n / \log n)$. Meanwhile, the flow time of $\OPT$ is $O(1)$ in both phases. Therefore, the competitive ratio is $\Omega(n / \log n)$.
\end{proof}

\subsection{Deterministic Algorithms: Multiple Machines without Kill-and-Restart}
\label{sec:multi-nonrestart-lb}
The lower bound $\Omega(n/m^2 + \sqrt{n/m})$  follows by combining our randomized lower bound of $\Omega(\sqrt{n/m})$ with the deterministic $\Omega(n/m^2)$ lower bound from~\cite{DBLP:journals/dam/EpsteinS06}, as presented below. 

\begin{theorem}
    \label{thm:no-restart-lb}
    Any deterministic online algorithm for total flow time minimization on parallel machines has a competitive ratio of at least $\Gamma(n,m)$, for every fixed pair $(n, m)$, where $\Gamma(n,m) = \Omega(n/m^2 + \sqrt{n/m})$.
\end{theorem}
\begin{proof}
    We first show that for every $n$, $m$, we can construct an adversary to prove $\Omega(n/m^2)$, the same as \cite{DBLP:journals/dam/EpsteinS06}. We still present it for completeness. 
    Consider the following case for any deterministic algorithm $\ALG$:
    \begin{itemize}
        \item A job with processing time $1$ is released at time $0$.
        \item Once $\ALG$ starts a job at $t$, the adversary releases $n$ jobs with processing time $1/n$ in $n / m$ batches, where each batch contains $m$ jobs released at time $t + 1/n, t + 2/n, \dots, t + 1/m$ respectively.
    \end{itemize}
    $\OPT$ can start all small jobs at their release time immediately; it can also solve the first job at $0$ if $t > 1$, or solve the first job at $2$ otherwise. Hence, $\OPT$ is $O(1)$ in both cases. For $\ALG$, as it can only use $m - 1$ machines to solve the small jobs, at most $m - 1$ jobs can be solved between the releases of two small job batches. The total flow time for $\ALG$ is at least $\sum_{i = 1}^{n / m} i/n = \Omega(n / m ^ 2)$. Therefore, the competitive ratio is $\Omega(n / m ^ 2)$.

    Finally, combining with our $\Omega(\sqrt{n/m})$ lower bound for all $n$, $m$, it proves the lemma. 
\end{proof}

\section{Lower Bounds for Algorithms without the Knowledge of \texorpdfstring{$n$}{n}}
\label{sec:unknownn_lower}

In this section, we establish strong lower bounds for the ``unknown $n$'' setting by the construction of a powerful \emph{oblivious adversary}. Unlike an adaptive adversary, this adversary does not react to an algorithm's random choices. Instead, it operates by analyzing an algorithm's description to derive its probability distribution for a single, initial job released at time $0$. Based on this distribution, the adversary classifies the algorithm into one of two categories: those that tend to start the job early, and those that tend to delay it. For each category, the adversary then commits to a specific instance to exploit that particular behavior.

\begin{theorem}
\label{thm:unknown-single-lb}
No randomized algorithm can achieve a competitive ratio of $o(n)$ without prior knowledge of $n$ in the single-machine setting.  
\end{theorem}

\begin{proof}
Recall that if a randomized algorithm $\ALG$ could achieve a competitive ratio of $o(n)$, then for every constant $C > 0$, there would exist a positive integer $n_0$ such that 
\[
\frac{\mathbb{E}[\ALG]}{\OPT} \le C \cdot n, \quad \text{for all } n \ge n_0.
\]
We now prove a contradiction: for any randomized algorithm $\ALG$, there exists a constant $C > 0$ such that for every positive integer $n_0$, one can construct an instance $I$ with $|I| = n \ge n_0$ satisfying
\[
\frac{\mathbb{E}[\ALG]}{\OPT} > C \cdot n.
\]
Hence, no randomized algorithm can achieve a competitive ratio of $o(n)$.

\medskip
\noindent
\textbf{Adversary construction.}
We construct an \emph{oblivious adversary} as follows. At time $0$, the adversary releases a single job of size $2$. Let $p(t)$ denote the probability density that $\ALG$ starts this job at time $t$, assuming no other jobs are present. Since $p(t)$ depends only on the algorithm's description, it is known to the adversary in advance. The algorithm must fall into one of the following categories (the parameters $x$ and $y$ will be fixed later):

\begin{enumerate}
    \item[\textbf{A.}] There exists an interval $[t, t+1)$, where $t \in \{1,2,\dots,x\}$, such that 
    \[
    \int_{t}^{t+1} p(s)\,ds \ge y.
    \]
    In this case, the adversary releases $n_1$ dummy jobs of size $0$ at time $t+1$. 
    With probability at least $y$, $\ALG$ eventually starts the size-$2$ job within $[t,t+1)$, all these $n_1$ jobs complete no earlier than time $t+2$, so each contributes at least $1$ to flow time and the algorithm’s total flow time is at least $n_1$, while an optimal solution completes all jobs by time $4$, i.e., $\OPT \le 4$. 
    Therefore,
    \[
    \frac{\mathbb{E}[\ALG]}{n \cdot \OPT} > \frac{y n_1}{4 (n_1 + 1)} \to \frac{y}{4}, \quad \text{as } n_1 \to \infty.
    \]
    \item[\textbf{B.}] For all $t \in \{1,2,\dots,x\}$, it holds that $\int_{t}^{t+1} p(s)\,ds < y$. 
    In this case, the probability that the size-$2$ job has not been started by time $x$ is at least $1 - xy$. 
    The adversary releases $n_0$ dummy jobs of size $0$ at time $x + 1$. 
    If $\ALG$ delays the size-$2$ job until after time $x$, its flow time alone is at least $x + 2$, while $\OPT \le 2$.
    Therefore,
    \[
    \frac{\mathbb{E}[\ALG]}{n \cdot \OPT} > \frac{(1 - xy) x}{2 (n_0 + 1)}.
    \]
\end{enumerate}

\paragraph{Parameter setting.}
Let $x = n_0$ and $y = \frac{1}{2n_0}$, so that $xy = \frac{1}{2}$. 
Then, in the Type~A case, there exists an instance $I$ with $|I| = n \gg n_0$ such that $\mathbb{E}[\ALG]/\OPT > (1/(8n_0)) \cdot n$;  
in the Type~B case, there exists an instance $I$ with $|I| = n = n_0 + 1 > n_0$ such that $\mathbb{E}[\ALG]/\OPT > (1/8) \cdot n$. 

Therefore, for every constant $0 < C \leq 1/8$, there exists a corresponding integer $n_0 = 1/(8C)$ such that
\[
\frac{\mathbb{E}[\ALG]}{\OPT} > C \cdot n, 
\]
for some instance with $n \ge n_0$.  
Consequently, no randomized algorithm can achieve a competitive ratio of $o(n)$.
\end{proof}

\begin{theorem}
\label{thm:unknown-multi-lb}
No randomized algorithm can achieve a competitive ratio of $o(n/m^2 + \sqrt{n/m})$ without prior knowledge of $n$ in the multiple-machine setting.
\end{theorem}

\noindent
The $\sqrt{n/m}$ term follows from the known-$n$ lower bound (\Cref{thm:multi-lb}). For the $n/m^2$ term, the intuition resembles \Cref{lem:opt-m-m-k}: losing one machine out of $m$ incurs an $\Omega(n/m^2)$ penalty in flow time. 

\begin{proof}
Since the $\Omega(\sqrt{n/m})$ lower bound holds even when $n$ is known (\Cref{thm:multi-lb}), it suffices to construct a hard instance achieving ratio $\Omega(n/m^2)$.

\medskip
\noindent
\textbf{Adversary construction.}
We use the same setup as in the single-machine case: the adversary first releases a size-2 job at time $0$ and classifies the algorithm into Type~A or Type~B based on the density $p(t)$. The only difference lies in how we construct the hard instance for Type~A.

\begin{enumerate}
    \item[\textbf{A.}] If $\ALG$ is of Type~A, the adversary releases $n_1$ additional jobs of size $1/k$, where $k = n_1/m$. The release pattern is:
    $$
    t+1,\; t+1+\frac{1}{k},\; t+1+\frac{2}{k},\; \dots,\; t+1+\frac{k-1}{k},
    $$
    and at each of these $k$ release times, $m$ new jobs (one per machine) are released. Hence, the total number of small jobs is $n_1 = m \cdot k$.

    The optimal scheduler can avoid any conflict between the size-2 job and the small jobs, so that all small jobs start immediately at their release times. The total flow time is therefore at most $m + 4$. However, with probability at least $y$, the algorithm starts the size-2 job in $[t, t+1)$, leaving only $m-1$ machines for small jobs. In this case, the total flow time is at least
    $$
    m + k \cdot 1 = m + \frac{n_1}{m} \ge \frac{n_1}{m}.
    $$
    Therefore,
    $$
    \frac{\mathbb{E}[\ALG]}{\OPT} \ge \frac{y \cdot (n_1/m)}{m + 4} \ge \frac{y n_1}{5 m^2},
    $$
    where the last inequality uses $m \ge 1$. As $n_1 \to \infty$, this ratio approaches $y n / (5m^2)$.

    \item[\textbf{B.}] The construction and argument are identical to the single-machine case: if the algorithm delays the size-2 job beyond time $x$, the expected ratio satisfies
    $$
    \frac{\mathbb{E}[\ALG]}{\OPT} > \frac{(1 - xy)x}{2(n_0 + 1)}.
    $$
\end{enumerate}

\paragraph{Parameter setting.}
We use the same parameters as before, setting $x = n_0$ and $y = 1/(2n_0)$ so that $xy = 1/2$. Under this choice, the ratio for Type~A algorithms is at least $(1/(10n_0)) \cdot n/m^2$, while for Type~B algorithms it is at least $(1/8) \cdot n \ge (1/8) \cdot n/m^2$.

Therefore, for every $0 < C \leq 1/10$, there exists a corresponding integer $n_0 = 1/(10C) \ge 1$ such that
$$
\frac{\mathbb{E}[\ALG]}{\OPT} > C \cdot \frac{n}{m^2},
$$
for some instance with $n > n_0$. Consequently, no randomized algorithm can achieve a competitive ratio of $o(n/m^2)$.
\end{proof}

\bibliographystyle{alphadin}
\bibliography{ref}

\appendix
\section{Generalized Version of Non-Preemptive Shortest Job First}
\label{sec:SJF}

In this section, we analyze the performance of the Non-Preemptive Shortest Job First (\NSJF) algorithm in a generalized version, with the concept of \emph{blocking period}: for blocking vector $\vec{b}$, where each $b_i$ denotes the initial blocking time of machine $i$ --- that is, machine $i$ is unavailable for processing during the interval $[0, b_i)$, we aim to prove the following bound:
$$
F(\NSJF(\vec{b})) \leq F(\OPT(\vec{0})) + \frac{nB}{m} + 2n\tau.
$$
where $B = \sum_{i=1}^m b_i$.

Since the context is clear, we will refer to $\NSJF(\vec{b})$ simply as \NSJF and $\OPT(\vec{0})$ as \OPT. Due to the initial blocking times, the number of active machines may differ between \NSJF and \OPT. To capture this, we define a function $a(t)$ to be the number of active machines at time $t$ under \NSJF. We also define the total active power of \NSJF up to time $t$ as:
$$
A(t) = \int_0^t a(y) \, dy.
$$

For any processing time threshold $p$ and time $t$, let $J_{\leq p}(t)$ be the set of jobs with $p_j \leq p$ and release time $r_j \leq t$.

\paragraph{Step 1: Volume Matching.}
For each algorithm, we track the volume of work associated with jobs in $J_{\leq p}(t)$. Let $V_{\leq p}(t)$ denote the total volume of work processed on jobs in $J_{\leq p}(t)$ by time $t$ under \OPT, and let $V'_{\leq p}(t)$ denote the total volume of jobs in $J_{\leq p}(t)$ that have been started by time $t$ under \NSJF.

\begin{lemma}[Volume Matching]
\label{lem:sjfvolume}
For every $t \geq 0$, define $t' = A^{-1}(m(t + \tau))$. Then for every threshold $p$, we have
$V'_{\leq p}(t') \geq V_{\leq p}(t)$.
\end{lemma}

\begin{proof}
Fix an arbitrary $p$, and let $t_0$ be the last time some machine is idle under $\NSJF$ before $t'$. Let $j$ be the job with $p_j > p$ that is started latest by \NSJF within the interval $(t_0, t')$. If no such job exists, we let $s_j = t_0$ for a fictitious job $j$ of infinite size.

Since \NSJF starts job $j$ with $p_j > p$ at time $s_j$, it must have already started all jobs with size at most $p$ released up to $s_j$ (due to its SJF policy). Thus, we have:
$$
    V'_{\leq p}(s_j) = \sum_{h \in J_{\leq p}(s_j)} p_h \geq V_{\leq p}(s_j).
$$

We now analyze two cases. If $s_j \geq t$, we immediately have:
$$
    V'_{\leq p}(t') \geq V'_{\leq p}(s_j) \geq V_{\leq p}(s_j) \geq V_{\leq p}(t).
$$
Otherwise, if $s_j < t$, then since $s_j$ is the latest start of a job with $p_j > p$ after the last idle time $t_0$, all processing in $[s_j + \tau, t')$ is on jobs of size at most $p$, and no machine is idle in this interval. Therefore:
\begin{align*}
    V'_{\leq p}(t')
    &\geq V'_{\leq p}(s_j) + \int_{s_j + \tau}^{t'} a(y) \, dy \\
    &\geq V_{\leq p}(s_j) + A(t') - A(s_j + \tau) \\
    &\geq V_{\leq p}(s_j) + m(t + \tau) - m(s_j + \tau) \\
    &\geq V_{\leq p}(s_j) + m(t - s_j) \\
    &\geq V_{\leq p}(t).
\end{align*}
\end{proof}

\paragraph{Step 2: Completion Matching.}
Let $N_{\leq p}(t)$ and $N'_{\leq p}(t)$ denote the number of jobs in $J_{\leq p}(t)$ completed by time $t$ under \OPT and \NSJF, respectively. When the subscript ${\leq p}$ is omitted, we count all completed jobs regardless of size.

\begin{lemma}[Completion Matching]
\label{lem:sjfnumber}
    For every $t \geq 0$, define $t' = A^{-1}(m(t+2\tau))$. For every threshold $p$, we have $N'_{\leq p}(t') \geq N_{\leq p}(t)$.
\end{lemma}
\begin{proof}
Define $t_1 = A^{-1}(m(t + \tau))$.

For each rank $i$, let $p_i$ be the size of the $i$-th smallest job completed under \OPT by time $t$, and let $p'_i$ be the size of the $i$-th smallest job started under \NSJF by time $t_1$. Let $k$ and $k'$ denote the total number of such jobs, respectively. We claim that $k' \ge k$ and $p'_i \le p_i$ for all $i \le k$.

Assume for contradiction that $i \le k$ is the smallest rank such that either $i > k'$ or $p'_i > p_i$. For \OPT, at least $i$ jobs of size at most $p_i$ are completed by time $t$, so:
$$
    V_{\leq p_i}(t) \geq \sum_{j=1}^{i} p_j.
$$
For \NSJF, all started jobs of size $\leq p_i$ are among the first $i - 1$ jobs (either because only $i - 1$ jobs have started, or because the $i$-th smallest is strictly larger than $p_i$). Therefore:
$$
    V'_{\leq p_i}(t_1) \leq \sum_{j=1}^{i-1} p'_j \leq \sum_{j=1}^{i-1} p_j,
$$
where the second inequality uses $p'_j \leq p_j$ for all $j < i$. This gives $V'_{\leq p_i}(t_1) < V_{\leq p_i}(t)$, contradicting \Cref{lem:sjfvolume}.

Therefore $k' \ge k$ and $p'_i \leq p_i$ for all $i \le k$, which implies that for every threshold $p$, the number of jobs of size $\leq p$ started by \NSJF by time $t_1$ is at least the number completed by \OPT by time $t$. Moreover, since the number of active machines in \NSJF is non-decreasing over time, all jobs started by time $t_1$ are completed by time $t'$. This completes the proof.
\end{proof}

\paragraph{Step 3: Bounding the Flow Time.}

\begin{lemma}
\label{lem:SJFgeneral}
Consider running \NSJF on jobs with processing times bounded by $\tau$, even in the presence of an initial blocking vector $\vec{b}$ on machines, where the total blocking time is $B = \sum_{i=1}^m b_i$. Then, compared to an offline preemptive optimal solution, even with migration, and without any blocking, the flow time of \NSJF is bounded by:
$$
    F(\NSJF(\vec{b})) \leq F(\OPT(\vec{0})) + 2n\tau + \frac{nB}{m}.
$$
\end{lemma}

\begin{proof}
Let $u$ denote the maximum completion time under \NSJF. The total flow times can be expressed as:
$$
    F(\OPT) = \int_0^\infty \left( |J(t)| - N(t) \right) \, dt, \quad
    F(\NSJF) = \int_0^u \left( |J(t)| - N'(t) \right) \, dt,
$$
where $|J(t)|$ is the number of jobs released by time $t$.

By \Cref{lem:sjfnumber}, \NSJF dominates \OPT after a shifted time point, specifically after $A^{-1}(m(t+2\tau))$. To bound the shift, observe that for any $\theta$,
$$
    A(\theta) = \int_0^\theta a(y) \, dy \geq m\theta - B,
$$
which implies:
$$
    A^{-1}(m(t+2\tau)) \leq t + 2\tau + \frac{B}{m}.
$$

Using this, we bound the flow time of \NSJF by shifting the completion curve by at most $2\tau + \frac{B}{m}$:
\begin{align*}
    F(\NSJF)
    &= \int_0^u \left( |J(t)| - N'(t) \right) \, dt \\
    &\leq \int_0^{u - 2\tau - \frac{B}{m}} \left( |J(t)| - N'(t + 2\tau + \frac{B}{m}) \right) \, dt + \int_{u - 2\tau - \frac{B}{m}}^u |J(t)| \, dt \\
    &\leq \int_0^{u - 2\tau - \frac{B}{m}} \left( |J(t)| - N(t) \right) \, dt + 2n\tau + \frac{nB}{m} \\
    &\leq F(\OPT) + 2n\tau + \frac{nB}{m}.
\end{align*}
\end{proof}

\section{Randomized Non-Preemptive Algorithm: Multiple Machines}
\label{sec:multi-random}

In this section, we extend \Cref{alg:dynamic} from the single-machine case to the multi-machine case. The intuition is introduced in \Cref{sec:single-random}. We detail the dynamic but online stable algorithm in \Cref{alg:multi-dynamic}.

\begin{algorithm}[ht]
    \caption{Dynamic (Online Stable) Non-Preemptive Randomized Algorithm for Multi-Machine}
    \label{alg:multi-dynamic}
    \SetKwFunction{FPartition}{Partition}
    \SetKwInOut{Input}{Include}  
    \Input{\Cref{alg:partition}}
    \KwData{Every job's current states as determined by \Cref{alg:partition}.}
    \SetKwBlock{OnJobRelease}{\textbf{On job $i$ release:}}{}
    \SetKwBlock{OnEvent}{\textbf{On any job completion or after any job release:}}{}
    \OnJobRelease
    {
        Call \FPartition{$i, \ell$} with $\ell = \lfloor \sqrt{nm} \rfloor$ to classify job $i$ and update related jobs\; 
        \If{$i$ is classified as large} {
            Sample $w_i \sim \mathrm{Unif}\{1, 2, \dotsc, \lfloor \sqrt{n/m} \rfloor \}$ \; 
            Sample $m_i \sim \mathrm{Unif}\{1, 2, \dotsc, m\}$ \; 
        }
        Run $\NSJF$ on all small jobs to get schedule $\S_1$ \tcp*{Including proxy jobs.}
        $\S_2 \gets \S_1$ \; 
        \For{each unproxied (active or committed) large job $j$ in order of release time $r_j$ } {
            $t \gets $ the first time in $\S_2$ where the cumulative idle time on machine $m_j$ since $r_j$ is at least $w_j p_j$\;
            Insert $j$ into $\S_2$ to start at $t$ on $m_j$ \tcp*{Overlapping jobs are shifted right.}
        }
        Schedule jobs according to $\S_2$\;
    }
\end{algorithm}

\begin{theorem}
\label{thm:random-no-restart-multi-alg}
    \Cref{alg:multi-dynamic} is an online polynomial-time randomized non-preemptive algorithm for total flow time minimization that is $O(\sqrt{n/m})$-competitive against the preemptive offline solution.
\end{theorem}
\begin{proof}
The algorithm is online stable and runs in polynomial time, for the same reason as in \Cref{lem:single-online-simulatable}. We use $S$ and $L$ to denote the sets of small and large jobs, where $S$ also includes proxy jobs. We use $J$ to denote the real job set, which includes only actual jobs and excludes proxy jobs.

For small jobs, by \Cref{lem:num_partition}, and \Cref{lem:sjfmain}, we have:
\[
F(\S_1) = O(\sqrt{n/m}) \cdot F(\OPT(S)) = O(\sqrt{n/m}) \cdot \OPT.
\]

For large jobs, we bound both the external delay and the self-delay. The external delay $\delta_j$ of a large job $j$ (i.e., the flow time increase caused to earlier jobs by inserting $j$) satisfies $\E[\delta_j] = O(\sqrt{n/m}) \cdot p_j$. The proof is as follows: Let $D_{k,l}$ be the set of jobs delayed by $j$ when $w_j = k$, $m_j = l$, where $k \in \{1, \dots, \lfloor \sqrt{n/m} \rfloor \}$ and $l \in \{1, \dots, m\}$. Since there is $p_j$ idle time between two possible adjacent insertion locations, all $D_{k,l}$ are disjoint. Then:
\[
\E[|D_{w_j,m_j}|] = \sum_{k=1}^{\lfloor \sqrt{n/m} \rfloor } \sum_{l=1}^{m} |D_{k,l}| \cdot \Pr[w_j = k] \cdot \Pr[m_j = l] = O(\sqrt{n/m}),
\]
so $\E[\delta_j] = O(\sqrt{n/m}) \cdot p_j$.

For the self-delay, if $j$ is proxied by $j'$, we count $r_{j'} - r_j$; otherwise, we count $\hat{C}_j - r_j$, where $\hat{C}_j$ is the completion time of $j$ immediately after insertion. This self-delay is denoted by $\hat{F}_j$. 
Ignoring the final $p_j$ processing time, the delay term only comes from collecting idle time or passing through busy time.

Each job collects at most $\sqrt{n/m} \cdot p_j$ units of idle time, so the total contribution from this is:
\[
\sqrt{n/m} \cdot \sum_{j \in L} p_j \leq \sqrt{n/m} \cdot \OPT.
\]

For busy time, each unit of busy period can be charged to at most $\sqrt{nm}$ jobs (by \Cref{lem:num_partition}), and the total busy time is at most $\OPT/m$, so (putting back $p_j$ for each $j$):
\[
\sum_{j\in L} \hat{F}_j \leq \sum_{j \in L} p_j + \sqrt{nm} \cdot \frac{\OPT}{m} = O(\sqrt{n/m}) \cdot \OPT.
\]

Putting everything together:
\[
F(\S_2) 
\leq F(\S_1) + \sum_{j \in L} \E[\delta_j] + \sum_{j \in L} \hat{F}_j = O(\sqrt{n/m}) \cdot \OPT.
\]
\end{proof}

\end{document}